\documentclass[journal]{IEEEtran}
\usepackage{cite}
\usepackage{amsmath,amssymb,amsfonts,bm}
\usepackage{graphicx}
\usepackage{booktabs}
\usepackage{multirow}
\usepackage{array}
\usepackage{algorithm}
\usepackage{algorithmic}
\usepackage[caption=false,font=footnotesize]{subfig}
\usepackage{url}
\usepackage{xcolor}
\usepackage{balance}
\usepackage[hidelinks]{hyperref}
\graphicspath{{figures/}}
\newtheorem{proposition}{Proposition}
\newtheorem{remark}{Remark}

\begin{document}

\title{Sparse Channel Acquisition in Aging Fluid Antenna Systems: A Reliability-Calibrated Posterior Reconstruction Approach}

\author{Xiaokai Song,
	Jie Tang, 
	Tuo Wu, 
	Geyang Xu,
	Zhongqing Wu,
	Baiyang Liu,
	Kin-Fai Tong,~\IEEEmembership{Fellow,~IEEE},\\
	and Naofal Al-Dhahir, \IEEEmembership{Fellow,~IEEE}
	
	\thanks{Xiaokai Song, Jie Tang, Tuo Wu, Geyang Xu, and Zhongqing Wu are with the School of Electronic and Information Engineering, South China University of Technology, Guangzhou 510640, China (e-mail: $\rm  songxiaokai@scut.edu.cn; eejtang@scut.edu.cn; wutuo@scut.edu.cn; \\ 202620111845@mail.scut.edu.cn; eezhongqing\_wu@mail.scut.edu.cn$).}
	\thanks{Baiyang Liu is with the School of Artificial Intelligence, Shenzhen Technology University, Shenzhen, China (e-mail:$\rm liubaiyang@sztu.edu.cn$).}
	\thanks{K.-F. Tong is with the School of Science and Technology, Hong Kong Metropolitan University, Hong Kong, China (e-mail:$\rm ktong@hkmu.edu.hk$).}
	\thanks{Naofal Al-Dhahir is with the Department of Electrical and Computer Engineering, The University of Texas at Dallas, Richardson, TX 75080 USA (e-mail: $\rm aldhahir@utdallas.edu$).}
}
\markboth{}%
{Song \MakeLowercase{\textit{et al.}}: Sparse Channel Acquisition in Aging Fluid Antenna Systems}

\maketitle

\begin{abstract}
Fluid antenna systems (FASs) exploit port-domain diversity within a compact aperture, but their gain can be lost if the receiver has to scan too many ports before data transmission. This paper studies sparse channel acquisition for aging FAS receivers, where the available information has mixed reliability: current observations are fresh but sparse, whereas a historical channel state information (CSI) is dense but stale due to Doppler-induced temporal decorrelation. We formulate the problem from a communication perspective by linking aperture-field reconstruction to port-selection rate, outage, probing budget, and online inference latency. We then cast the acquisition task as reliability-calibrated posterior reconstruction, in which the sparse current observations serve as hard measurement evidence and historical CSI is reused only as soft temporal side information. 
The resulting reliability-controlled conditional diffusion posterior (DCP) framework conditions the denoising process on the sparse current observations, the observation mask, and the historical CSI, and refines the reverse trajectory through observation consistency, mismatch-aware temporal gating, spatial trust weighting, progressive temporal activation, and warm initialization.
Analysis shows how aging and reconstruction error induce rate distortion and rate regret. Simulations across probing budgets, signal-to-noise ratios (SNRs), temporal correlations, observation patterns, ablation settings, and latency tests show consistent reconstruction gains with millisecond-level online inference latency, indicating that stale CSI can be beneficial when its strength, timing, and spatial support are explicitly controlled.
\end{abstract}


\section{Introduction}
\IEEEPARstart{F}{luid} antenna systems (FASs) exploit fine-grained spatial variation within a compact aperture through reconfigurable positions or radiation states, avoiding a large number of simultaneously active radio-frequency chains \cite{Lai2026JSAC,Shojaeifard2022ProcIEEE}. Their diversity and outage benefits are well established \cite{Wong2021TWC,New2024Outage,New2025Tutorial}, while recent studies emphasize aperture, correlation, switching, and hardware constraints \cite{Wu2025MWC,Lu2025MCOM,Meta-Fluid lot}. A central practical question is whether the aperture field can
be acquired accurately under practical acquisition overhead
constraints \cite{Ramirez2024TWC,Zhang2025SBAR}.

This issue is central to port selection and FAMA: exhaustive scanning improves channel knowledge but increases switching delay, whereas sparse probing reduces latency while leaving most ports unobserved \cite{New2026JSAC,Chai2022CL,Wong2022TWC,Wong2023SlowFAMA,Wong2023FastFAMA}. Learning-aided FAMA and cGAN-based selection reduce scanning through port correlation \cite{Waqar2023CL,Eskandari2024WCL}, but still depend on reliable inference of unobserved or outdated ports.

The same bottleneck appears in broader FAS and reconfigurable-aperture designs, but with different system objectives. FAS-CoNOMA uses port-domain flexibility together with code-domain multiple access, so the receiver must resolve both spatial and code-domain structure \cite{Wu2026CoNOMA}. Secrecy-oriented and scalable FAS studies show that port correlation and aperture modeling affect not only selection gain but also security and array-signal-processing behavior \cite{Wu2026Secrecy,Wu2026ScalableFAS}. Related work on artificial intelligence (AI)-native FAS and movable antennas reaches a similar conclusion from the learning and hardware perspectives: high-dimensional spatial information is valuable only when it can be represented and acquired under realistic overhead constraints \cite{Wang2024AIFAS,Zhu2024MAopportunities}. Thus, the acquisition problem is not a secondary implementation detail; it is part of the system model that determines whether port-domain flexibility can be effectively
exploited.

Aperture-field reconstruction addresses this problem by estimating the aperture field from incomplete measurements. Oversampling has been shown to be essential for FAS channel estimation and reconstruction, highlighting that dense-port acquisition cannot be treated as ordinary low-dimensional channel state information (CSI) estimation
\cite{New2025TWC}. Channel estimation for FAS-assisted multiuser millimeter-wave (mmWave) systems has also been investigated under structured propagation environments, where sparse measurements require structure-aware interpretation \cite{Xu2024CL}. More recently, sparsity-driven inverse scattering has been used to reconstruct spatial information from limited measurements in meta-fluid antenna systems \cite{Meta-Fluid}. A successive Bayesian reconstructor has further been developed, demonstrating the value of probabilistic recovery for FAS channel estimation \cite{Zhang2025SBAR}. Sparse Bayesian learning has also been applied directly to FAS estimation, further confirming that model-based sparse recovery is a natural baseline for this task \cite{Xu2025WCL}. These works establish the importance of reconstruction, but most of them focus on spatial recovery from current observations rather than on the joint use of sparse current observations and stale CSI.

Classical estimators remain useful but model dependent: linear estimation relies on reliable covariance and pilot design \cite{Hassibi2003,Biguesh2006}, whereas compressed sensing and SBL depend on dictionary or sparsity-model match \cite{Donoho2006,Tipping2001}.
Related large-scale (XL) array studies further demonstrate the value of
structure-aware recovery by exploiting sparsity induced by
dual-wideband effects, near-field propagation, and spatial
non-stationarity for Bayesian channel estimation \cite{NEW-XL-MIMO}, and hybrid spherical--planar structure associated with modular array geometry for low-complexity channel-parameter estimation and three-dimensional localization \cite{NEW-MODULAR-XL}. 
Gaussian process regression offers a principled way to interpolate correlated spatial samples, yet it relies on a kernel that may not capture both aperture-scale smoothness and local multipath-induced texture \cite{Rasmussen2006}. Learning-based reconstruction relaxes some of these assumptions. Graph masked autoencoders have been used for FAS channel extrapolation \cite{Haibin2024WCL}, while score-based and diffusion-based estimators can represent a distribution of plausible wireless channels instead of a single interpolation rule \cite{Arvinte2023TWC,Chen2025TCOM}. The diffusion literature further provides practical denoising
procedures and denoising diffusion implicit model (DDIM) sampling tools for such posterior-style generation \cite{Ho2020,Song2020}. However, these methods are usually built around either a synchronized sparse snapshot or side information that is assumed to be reliable.

Channel aging breaks the reliable-side-information assumption. Temporal correlation decays with Doppler and delay \cite{Clarke1968,Jakes1974}, consistent with time-evolving geometry-based models \cite{Jaeckel2014}. Thus, an aging FAS receiver combines fresh but sparse current observations with dense but potentially stale historical CSI, whose mismatch can bias unobserved ports. 
Hence, the receiver faces a reliability-calibrated inference problem with  a fresh but sparse current observation and dense but potentially mismatched historical CSI.

This paper studies that inference problem and proposes a reliability-controlled conditional diffusion posterior (DCP) framework. 
The main idea is to reconstruct the current aperture field by giving priority to sparse current observations at observed ports, while allowing historical CSI to guide only where it remains credible.
DCP conditions the denoising process on the sparse current observations, the observation mask, and the historical CSI, and controls its reverse trajectory through observation consistency, mismatch-driven temporal gating, spatial trust weighting, progressive temporal activation, and warm initialization. The sparse current observations anchor the posterior, and stale CSI is reused only under reliability control.

The main contributions are summarized as follows.
\begin{itemize}
    \item We develop a communication-oriented model for sparse acquisition in aging FAS receivers. The formulation connects the aperture-field reconstruction problem to port-selection rate, outage probability, probing budget, and online inference latency.
    \item We formulate sparse acquisition with mixed-reliability CSI and propose DCP, which combines conditional diffusion reconstruction with observation-consistent updates, sample-wise temporal gating, spatial trust weighting, progressive temporal activation, and reliability-aware warm initialization.
    \item We provide communication-oriented and sampler-level analysis. The communication analysis links channel aging and reconstruction error to rate distortion and rate regret, while the sampler analysis shows how observation guidance contracts measured-port residuals and how adaptive temporal weighting limits stale-prior bias.
    \item We validate the framework through probing budgets, signal-to-noise ratios (SNRs), temporal-correlation, ablation settings, observation-pattern, and latency tests, highlighting both normalized mean-square error (NMSE) gains and the runtime implications for online aperture-field reconstruction.
    
\end{itemize}

The rest of this paper is organized as follows. Section~\ref{sec:model} presents the aging FAS model, communication-oriented acquisition objective, and posterior formulation. Section~\ref{sec:method} gives the DCP sampler and reliability-control mechanisms. Section~\ref{sec:theory} provides analytical support for the proposed design. Section~\ref{sec:setup} describes the data construction and baselines. Section~\ref{sec:results} reports the numerical results, and Section~\ref{sec:conclusion} concludes the paper.

\section{System Model and Problem Formulation}\label{sec:model}

\begin{table*}[t]
\centering
\caption{Core Notation Used in Sections II and III}
\label{tab:notation}
\small
\setlength{\tabcolsep}{3.5pt}
\renewcommand{\arraystretch}{1.08}
\begin{tabular}{@{}p{0.12\textwidth} p{0.35\textwidth} p{0.12\textwidth} p{0.35\textwidth}@{}}
\toprule
Symbol & Definition & Symbol & Definition \\
\midrule
$N_h,N_w$ & Numbers of FAS ports along the two aperture dimensions. &
$\mathbf{H}^{\mathrm{c}}_t,\mathbf{H}_t$ & Complex aperture field and its real--imaginary stacked form at slot $t$. \\
$\mathbf{H}^{\mathrm{old}}_t$ & Dense historical CSI available at slot $t$. &
$\rho_t$ & Slot-to-slot temporal correlation coefficient. \\
$d,\alpha_t$ & Historical-CSI age in slots and additional fidelity coefficient. &
$\mathbf{M}_t,\mathbf{Y}_t$ & Binary observation mask and sparse current observation. \\
$\eta_t$ & Probing budget ratio, $\|\mathbf{M}_t\|_0/(N_hN_w)$. &
$\bar{\gamma},R_{t,n}$ & Receive-SNR scaling factor and achievable rate of port $n$. \\
$n_t^\star,\widehat n_t$ & Oracle and reconstruction-based selected port indices. &
$\widehat{\mathbf{H}}_t$ & Reconstructed current aperture field. \\
$\mathbf{C}_t$ & DCP condition tensor formed from $\mathbf{H}^{\mathrm{old}}_t$, $\mathbf{Y}_t$, and $\mathbf{M}_t$. &
$\lambda_{\mathrm{obs}}$ & Observation-guidance weight for measured-port anchoring. \\
$\eta_t^{\mathrm{m}},g_t$ & Historical/current mismatch metric and sample-wise temporal gate. &
$\mathbf{T}_t$ & Spatial trust map for reliability-controlled historical-CSI reuse. \\
$r_k,r_0$ & Reverse-sampling progress ratio and temporal-activation threshold. &
$\lambda_{\mathrm{tmp}}^{(k)}$ & Progressive temporal-regularization weight at reverse step $k$. \\
$\xi_t$ & Gate-controlled warm-start coefficient. &
$\mathcal{K},S,L$ & Selected DDIM step set, retained reverse-step count, and posterior-sample count. \\
\bottomrule
\end{tabular}
\end{table*}

\subsection{Two-Dimensional FAS Aperture-Field Representation}

Consider a single-user receiving FAS whose physical aperture is discretized into an $N_h\times N_w$ grid of switchable positions. Let the aperture size be $W_h\lambda_{\mathrm{c}}\times W_w\lambda_{\mathrm{c}}$, where $\lambda_{\mathrm{c}}$ is the carrier wavelength. The location of the $(p,q)$th port is
\begin{equation}
\mathbf{r}_{p,q}=
\left[
\frac{pW_h\lambda_{\mathrm{c}}}{N_h-1},
\frac{qW_w\lambda_{\mathrm{c}}}{N_w-1}
\right]^{\mathrm{T}},
\label{eq:port_location}
\end{equation}
where $p=0,\ldots,N_h-1$ and $q=0,\ldots,N_w-1$. Let $\mathbf{H}^{\mathrm{c}}_t\in\mathbb{C}^{N_h\times N_w}$ denote the complex aperture field at slot $t$, whose $(p,q)$th entry is the small-scale fading coefficient observed at $\mathbf{r}_{p,q}$. For neural processing, we represent it by the real-imaginary stacked tensor
\begin{equation}
\mathbf{H}_t = \mathrm{Stack}\!\left(\Re\{\mathbf{H}^{\mathrm{c}}_t\},\Im\{\mathbf{H}^{\mathrm{c}}_t\}\right) \in \mathbb{R}^{2\times N_h \times N_w}.
\label{eq:Hdef}
\end{equation}
This representation is lossless and converts the complex aperture field into a real-valued tensor suitable for convolutional processing. The physical basis for sparse probing is the spatial correlation among nearby ports. Under the classical isotropic rich-scattering model, the second-order spatial covariance between two aperture locations $\mathbf{r}_i$ and $\mathbf{r}_j$ can be written as\cite{Clarke1968}
\begin{equation}
[\mathbf{R}_{\mathrm{s}}]_{i,j}
=\mathbb{E}\!\left[H_t^{\mathrm{c}}(\mathbf{r}_i)\left(H_t^{\mathrm{c}}(\mathbf{r}_j)\right)^*\right]
=\sigma_h^2 J_0\!\left(\frac{2\pi\|\mathbf{r}_i-\mathbf{r}_j\|_2}{\lambda_{\mathrm{c}}}\right),
\label{eq:spatialcorr}
\end{equation}
where $J_0(\cdot)$ is the zeroth-order Bessel function of the first kind and $\sigma_h^2$ is the average channel power. Equation~\eqref{eq:spatialcorr} is not required by the proposed algorithm, but it clarifies the wireless origin of the aperture-field structure: densely spaced ports do not represent statistically independent channel samples, and their correlation depends on their physical separation in wavelengths.

\subsection{Channel Aging and Historical CSI}
The aperture field evolves over time. We model its slot-to-slot dynamics as
\begin{equation}
\mathbf{H}^{\mathrm{c}}_t = \rho_t\mathbf{H}^{\mathrm{c}}_{t-1}+\sqrt{1-\rho_t^2}\,\mathbf{U}^{\mathrm{c}}_t,
\label{eq:aging}
\end{equation}
where $\rho_t\in[0,1]$ is the temporal correlation coefficient and $\mathbf{U}^{\mathrm{c}}_t$ is a zero-mean innovation field independent of $\mathbf{H}^{\mathrm{c}}_{t-1}$ with the same spatial covariance structure as the aperture field. For a slot duration $T_{\mathrm{s}}$ and maximum Doppler frequency $f_{\mathrm{D}}=v/\lambda_{\mathrm{c}}$, the Clarke-Jakes model gives the temporal correlation
\begin{equation}
\rho_t=\left|J_0\!\left(2\pi f_{\mathrm{D}}T_{\mathrm{s}}\right)\right|,
\label{eq:doppler_rho}
\end{equation}
where $v$ is the receiver speed. This relation makes the aging effect explicit: larger mobility or a longer acquisition delay reduces the usefulness of previously observed port channels. The first-order model in \eqref{eq:aging} is therefore used as a controlled statistical description of temporal mismatch rather than as an exact propagation law.

The available temporal side information is not assumed to equal the previous-slot field. Instead, we model the historical CSI as
\begin{equation}
\mathbf{H}^{\mathrm{old,c}}_t = \alpha_t\mathbf{H}^{\mathrm{c}}_{t-d}+\sqrt{1-\alpha_t^2}\,\mathbf{S}^{\mathrm{c}}_t+\mathbf{N}^{\mathrm{old,c}}_t,
\label{eq:hold}
\end{equation}
where $d\in\mathbb N^{+}$ denotes the age of the historical CSI in slots, such that $\mathbf H_{t-d}^{c}$ is the aperture field generated $d$ slots before the current slot. Here, $d$ determines the source slot of the historical CSI, while $\alpha_t$ is a historical-CSI fidelity coefficient that controls only the additional degradation caused by prediction, storage, scanning, calibration, or residual modeling uncertainty. The temporal decorrelation between $\mathbf H_t^{c}$ and $\mathbf H_{t-d}^{c}$ is already induced by the channel evolution in \eqref{eq:aging} and is not applied again through $\alpha_t$. $\mathbf{S}^{\mathrm{c}}_t$ is a structured perturbation term, and $\mathbf{N}^{\mathrm{old,c}}_t$ denotes historical CSI noise. Its real-imaginary stacked form is denoted by $\mathbf{H}^{\mathrm{old}}_t$. Thus, the historical CSI carries temporal information but may deviate from the current aperture field.

\begin{remark}
The receiver does not assume access to the current aperture field. It observes only the sparse current observations in \eqref{eq:obsmodel} and a dense historical CSI generated before the current slot. This distinction is crucial: the historical CSI can help infer unobserved ports, but using it as current ground truth would ignore Doppler aging and can bias the selected port.
\end{remark}

\subsection{Sparse Observation Model}
At slot $t$, only a small subset of ports can be probed. Let $\mathbf{M}_t\in\{0,1\}^{1\times N_h\times N_w}$ denote the binary observation mask. The resulting sparse observation is
\begin{equation}
\mathbf{Y}_t = \mathbf{M}_t\odot(\mathbf{H}_t+\mathbf{W}_t),
\label{eq:obsmodel}
\end{equation}
where $\odot$ denotes elementwise multiplication and $\mathbf{W}_t$ is observation noise. The probing budget ratio is defined as
\begin{equation}
\eta_t = \frac{\|\mathbf{M}_t\|_0}{N_hN_w}.
\label{eq:budget}
\end{equation}
The reconstruction task is to estimate $\mathbf{H}_t$ from $\mathbf{Y}_t$, $\mathbf{M}_t$, and $\mathbf{H}^{\mathrm{old}}_t$.

\subsection{Communication-Oriented Acquisition Objective}
The reason for reconstructing the aperture field is not merely to reduce an image-like error metric. In an FAS receiver, the reconstructed aperture field is used to activate a port for data transmission. Let $N=N_hN_w$, and let $h_{t,n}$ denote the complex channel coefficient at the $n$th port after vectorizing $\mathbf{H}^{\mathrm{c}}_t$. When port $n$ is activated, the received signal is
\begin{equation}
r_{t,n}=\sqrt{P}h_{t,n}s_t+v_{t,n},
\label{eq:rx_signal}
\end{equation}
where $P$ is the transmit power, $s_t$ is the unit-power data symbol, and $v_{t,n}\sim\mathcal{CN}(0,\sigma_v^2)$ is receiver noise. The corresponding receive SNR and achievable rate are
\begin{equation}
\gamma_{t,n}=\bar{\gamma}|h_{t,n}|^2,\qquad
R_{t,n}=\log_2(1+\gamma_{t,n}),
\label{eq:rate_def}
\end{equation}
where $\bar{\gamma}=P/\sigma_v^2$. If the current aperture field were known, the receiver would select
\begin{equation}
n_t^\star=\arg\max_{1\le n\le N}R_{t,n}
=\arg\max_{1\le n\le N}|h_{t,n}|^2 .
\label{eq:oracle_port}
\end{equation}
With sparse acquisition, however, the receiver must choose a port from imperfect information. A reconstruction rule $\mathcal{R}$ maps $(\mathbf{Y}_t,\mathbf{M}_t,\mathbf{H}^{\mathrm{old}}_t)$ to $\widehat{\mathbf{H}}_t$, and the activated port can be written as
\begin{equation}
\widehat{n}_t=\Pi(\widehat{\mathbf{H}}_t),
\label{eq:selected_port}
\end{equation}
where $\Pi(\cdot)$ is the port-selection rule. For maximum-rate selection from the reconstructed aperture field, $\Pi(\widehat{\mathbf{H}}_t)=\arg\max_n|\widehat{h}_{t,n}|^2$.

Hence, the communication objective is to choose an acquisition and reconstruction rule that preserves the utility of the selected port under pilot and latency constraints. A representative formulation is
\begin{align}
\text{(P0)}\quad
\min_{\mathcal{R},\Pi}\;&
\mathbb{E}\!\left[R_{t,n_t^\star}-R_{t,\widehat{n}_t}\right] \notag\\
\text{s.t.}\;&
\mathbf{P}\!\left(\gamma_{t,\widehat{n}_t}<\gamma_{\mathrm{th}}\right)\le \epsilon_{\mathrm{out}},\notag\\
&\|\mathbf{M}_t\|_0\le P_{\max},\qquad T_{\mathrm{inf}}\le T_{\max}.
\label{eq:comm_problem}
\end{align}
Here $\gamma_{\mathrm{th}}$ is the outage threshold, $P_{\max}$ is the probing budget, and $T_{\mathrm{inf}}$ is the online inference latency. Problem (P0) makes the role of reconstruction precise: normalized mean-square error is a useful surrogate, but the communication system ultimately cares about rate regret, outage, and whether the best-port region is preserved.

The next result quantifies why stale CSI cannot be used as if it were current CSI.
\begin{proposition}[Aging-induced rate distortion]\label{prop:aging_rate}
Assume that, for a fixed port, $h_t=\rho_d h_{t-d}+\sqrt{1-\rho_d^2}u_t$, where $\rho_d\in[0,1]$ denotes the effective temporal correlation between slots $t-d$ and $t$, and $h_{t-d}$ and $u_t$ are independent proper complex Gaussian variables with variance $\sigma_h^2$. Let $R_t=\log_2(1+\bar{\gamma}|h_t|^2)$ and $R_{t-d}=\log_2(1+\bar{\gamma}|h_{t-d}|^2)$. Then
\begin{equation}
\mathbb{E}\!\left[|R_t-R_{t-d}|\right]
\le
\frac{2\sqrt{2}\bar{\gamma}\sigma_h^2}{\ln 2}\sqrt{1-\rho_d}.
\label{eq:aging_rate_bound}
\end{equation}
\end{proposition}
\noindent\textit{Proof:}
Since $x\mapsto\log_2(1+\bar{\gamma}x)$ is $\bar{\gamma}/\ln2$-Lipschitz for $x\ge0$,
\begin{equation}
|R_t-R_{t-d}|
\le \frac{\bar{\gamma}}{\ln2}\left||h_t|^2-|h_{t-d}|^2\right|.
\end{equation}
Moreover,
\begin{equation}
\left||h_t|^2-|h_{t-d}|^2\right|
\le |h_t-h_{t-d}|(|h_t|+|h_{t-d}|).
\end{equation}
By the Cauchy--Schwarz inequality, $\mathbb{E}|h_t-h_{t-d}|^2=2\sigma_h^2(1-\rho_d)$ and $\mathbb{E}(|h_t|+|h_{t-d}|)^2\le4\sigma_h^2$. Combining these terms gives \eqref{eq:aging_rate_bound}. $\hfill\blacksquare$

Proposition~\ref{prop:aging_rate} shows that the communication distortion caused by stale CSI grows with temporal decorrelation. Thus, the historical CSI should be treated as a reliability-dependent memory, not as a second set of sparse current observations.

The following result connects reconstruction accuracy to the port-selection objective in (P0).
\begin{proposition}[Reconstruction-to-rate regret]\label{prop:rate_regret}
Let $\widehat{h}_{t,n}$ be the reconstructed channel at port $n$, and suppose that
\begin{equation}
\max_{1\le n\le N}|\widehat{h}_{t,n}-h_{t,n}|\le \varepsilon,\qquad
\max_{1\le n\le N}|h_{t,n}|\le A .
\end{equation}
If $\widehat{n}_t=\arg\max_n|\widehat{h}_{t,n}|^2$, then the rate regret satisfies
\begin{equation}
0\le R_{t,n_t^\star}-R_{t,\widehat{n}_t}
\le
\frac{4\bar{\gamma}}{\ln2}\left(A\varepsilon+\varepsilon^2\right).
\label{eq:rate_regret_bound}
\end{equation}
Furthermore, if the gain gap between the strongest and second-strongest true ports is $\delta_t$, exact top-port recovery is guaranteed whenever $\delta_t>2\varepsilon(2A+\varepsilon)$.
\end{proposition}
\noindent\textit{Proof:}
Let $a_n=|h_{t,n}|$ and $\widehat{a}_n=|\widehat{h}_{t,n}|$. Since $|\widehat{a}_n-a_n|\le\varepsilon$ and $\widehat{a}_{\widehat{n}_t}\ge\widehat{a}_{n_t^\star}$, we have $a_{n_t^\star}\le a_{\widehat{n}_t}+2\varepsilon$. Hence
\begin{equation}
a_{n_t^\star}^2-a_{\widehat{n}_t}^2
\le 4A\varepsilon+4\varepsilon^2.
\end{equation}
The Lipschitz property of $\log_2(1+\bar{\gamma}x)$ then yields \eqref{eq:rate_regret_bound}. For the exact-recovery statement, note that the gain-estimation error obeys
\begin{equation}
\left||\widehat{h}_{t,n}|^2-|h_{t,n}|^2\right|
\le \varepsilon(2A+\varepsilon),
\end{equation}
so the ordering between the strongest port and all other ports is preserved if the true gain gap exceeds twice this error bound. $\hfill\blacksquare$

Proposition~\ref{prop:rate_regret} explains why aperture-field reconstruction must preserve local peaks and not only reduce average error. It also motivates a posterior formulation: the receiver should combine the sparse current observations, historical CSI, and an aperture-field prior to reduce the uncertainty of the ports that can affect the final communication decision.

\subsection{Reliability-Calibrated Posterior Objective}
Rather than treating the task as one-shot deterministic regression, we formulate aperture-field reconstruction under sparse observations as reliability-calibrated posterior inference:
\begin{equation}
\widehat{\mathbf{H}}_t \in \mathcal{S}\left[p\left(\mathbf{H}_t\mid \mathbf{Y}_t,\mathbf{M}_t,\mathbf{H}^{\mathrm{old}}_t\right)\right],
\label{eq:posterior}
\end{equation}
where $\mathcal{S}[\cdot]$ denotes posterior sampling or a posterior-mean approximation. By Bayes' rule, the posterior can be factorized as
\begin{align}
p(\mathbf{H}_t|\mathbf{Y}_t,\mathbf{M}_t,\mathbf{H}^{\mathrm{old}}_t)
&\propto
p(\mathbf{Y}_t|\mathbf{H}_t,\mathbf{M}_t) \notag\\
&\quad\times
p(\mathbf{H}^{\mathrm{old}}_t|\mathbf{H}_t,\mathbf{M}_t)
p_\theta(\mathbf{H}_t),
\label{eq:posterior_factor}
\end{align}
where $p_\theta(\mathbf{H}_t)$ is the learned aperture-field prior. The observation likelihood follows from \eqref{eq:obsmodel}. For real-stacked aperture fields with noise variance $\sigma_w^2$, it can be written as
\begin{equation}
-\log p(\mathbf{Y}_t|\mathbf{H}_t,\mathbf{M}_t)
=\frac{1}{2\sigma_w^2}\left\|\mathbf{M}_t\odot(\mathbf{H}_t-\mathbf{Y}_t)\right\|_F^2+c_1.
\label{eq:obs_likelihood}
\end{equation}
The historical CSI is not imposed as a hard constraint. Its reliability is represented instead by a precision field $\bm{\Lambda}_t$, yielding the aging-aware temporal likelihood
\begin{align}
-\log p(\mathbf{H}^{\mathrm{old}}_t|\mathbf{H}_t,\mathbf{M}_t)
&=\frac{1}{2}\left\|(\mathbf{1}-\mathbf{M}_t)\odot\bm{\Lambda}_t^{1/2}\right. \notag\\
&\quad\left.\odot(\mathbf{H}_t-\mathbf{H}^{\mathrm{old}}_t)\right\|_F^2+c_2.
\label{eq:hist_likelihood}
\end{align}
The mask $\mathbf{1}-\mathbf{M}_t$ gives priority to sparse current observations at measured entries and allows the historical CSI to regularize primarily unmeasured entries. The negative log-posterior is then, up to constants,
\begin{align}
\mathcal{J}(\mathbf{H}_t)&=
\frac{1}{2\sigma_w^2}\left\|\mathbf{M}_t\odot(\mathbf{H}_t-\mathbf{Y}_t)\right\|_F^2 \notag\\
&\quad+\frac{1}{2}\left\|(\mathbf{1}-\mathbf{M}_t)\odot\bm{\Lambda}_t^{1/2}\right. \notag\\
&\quad\left.\odot(\mathbf{H}_t-\mathbf{H}^{\mathrm{old}}_t)\right\|_F^2
-\log p_\theta(\mathbf{H}_t).
\label{eq:posterior_energy}
\end{align}
This decomposition provides a direct interpretation of the proposed method: the diffusion model supplies a learned prior, the observation term enforces consistency with the sparse current observations, and the temporal term reuses stale CSI under reliability control.
\section{Aging-FAS Posterior Reconstruction Algorithm}\label{sec:method}
DCP solves the posterior problem in \eqref{eq:posterior_energy} under sparse probing, stale CSI, and latency constraints. Offline, a conditional diffusion model learns aperture structure from the same side information used at inference; online, observation guidance anchors probed ports while reliability-controlled temporal guidance reuses historical CSI on unobserved ports.

\subsection{FAS Posterior Decomposition Under Sparse Port Probing}
The reverse diffusion trajectory is interpreted as an approximate sampler of the posterior in \eqref{eq:posterior_factor}. Let $\mathbf{X}$ denote an intermediate clean aperture-field estimate at a reverse step. From \eqref{eq:posterior_energy}, the posterior score admits the decomposition
\begin{align}
\nabla_{\mathbf{X}}\log p(\mathbf{X}|\mathbf{Y}_t,\mathbf{M}_t,\mathbf{H}^{\mathrm{old}}_t)
&=\nabla_{\mathbf{X}}\log p_\theta(\mathbf{X}) \notag\\
& -\frac{1}{\sigma_w^2}\mathbf{M}_t\odot(\mathbf{X}-\mathbf{Y}_t) \notag\\
& -(\mathbf{1}-\mathbf{M}_t)\odot\bm{\Lambda}_t
\odot(\mathbf{X}-\mathbf{H}^{\mathrm{old}}_t).
\label{eq:score_decomp}
\end{align}
The first term is learned by the diffusion denoiser and captures the aperture-field structure that links unobserved ports to observed ones. The second pulls the sample toward the ports measured in the current slot, and the third reuses the historical CSI according to an aging-dependent precision. In the FAS setting, these three terms correspond to spatial aperture regularity, fresh but sparse current observations, and dense but potentially outdated historical CSI.

To connect the precision field with observable quantities, we set
\begin{equation}
\bm{\Lambda}_t = \lambda_{\mathrm{tmp}}^{(k)} g_t \mathbf{T}_t,
\label{eq:lambda_field}
\end{equation}
where $g_t$ controls sample-wise temporal reliability and $\mathbf{T}_t$ controls spatially varying trust.
In implementation, a small mismatch between the historical CSI and the currently observed entries produces a large effective temporal weight, whereas a large mismatch suppresses that weight. This design uses the overlap between historical CSI and currently observed ports to decide whether the historical aperture peaks remain informative for the current slot.

\subsection{Aperture-Field Prior Conditioned on FAS Side Information}
The learned prior must preserve both aperture-scale correlation and local peaks relevant to port activation. We therefore use a conditional diffusion model whose denoiser receives the sparse current observations, historical CSI, and observation mask, explicitly distinguishing measured, missing, and temporal side information.

The condition tensor is constructed as
\begin{equation}
\mathbf{C}_t = \mathrm{Concat}\left(\mathbf{H}^{\mathrm{old}}_t,\mathbf{Y}_t,\mathbf{M}_t\right),
\label{eq:condinput}
\end{equation}
where the historical CSI and current observation each use two real-imaginary channels, while the mask uses one channel. Hence, the condition tensor provides historical CSI, sparse current observations, and measurement support.

In the forward diffusion process,
\begin{equation}
\mathbf{H}_{t,k}=\sqrt{\bar{\alpha}_k}\,\mathbf{H}_t+\sqrt{1-\bar{\alpha}_k}\,\bm{\epsilon},\qquad \bm{\epsilon}\sim\mathcal{N}(\mathbf{0},\mathbf{I}),
\label{eq:forwarddiff}
\end{equation}
where $k$ is the diffusion step and $\bar{\alpha}_k=\prod_{i=1}^k(1-\beta_i)$. The network $\bm{\epsilon}_\theta(\cdot)$ is trained by minimizing
\begin{equation}
\mathcal{L}_{\mathrm{cond}} = \mathbb{E}_{t,k,\bm{\epsilon}}\Big[\|\bm{\epsilon}-\bm{\epsilon}_\theta(\mathbf{H}_{t,k},k,\mathbf{C}_t)\|_F^2\Big].
\label{eq:lcond}
\end{equation}
This standard noise-prediction objective aligns offline learning with online aperture-field reconstruction by exposing the denoiser to the same side information available during inference. Operationally, it learns how complete aperture fields are distributed when only a subset of ports is freshly measured and the dense side information may be stale.

The denoiser is a two-dimensional U-Net-style encoder--decoder with time embeddings and skip connections \cite{Ronneberger2015}. As shown in Fig.~\ref{fig:network}, conditioning is supplied at every reverse step to preserve measured values and decision-relevant local aperture structure.

\begin{figure*}[t]
    \centering
    \includegraphics[width=0.92\textwidth]{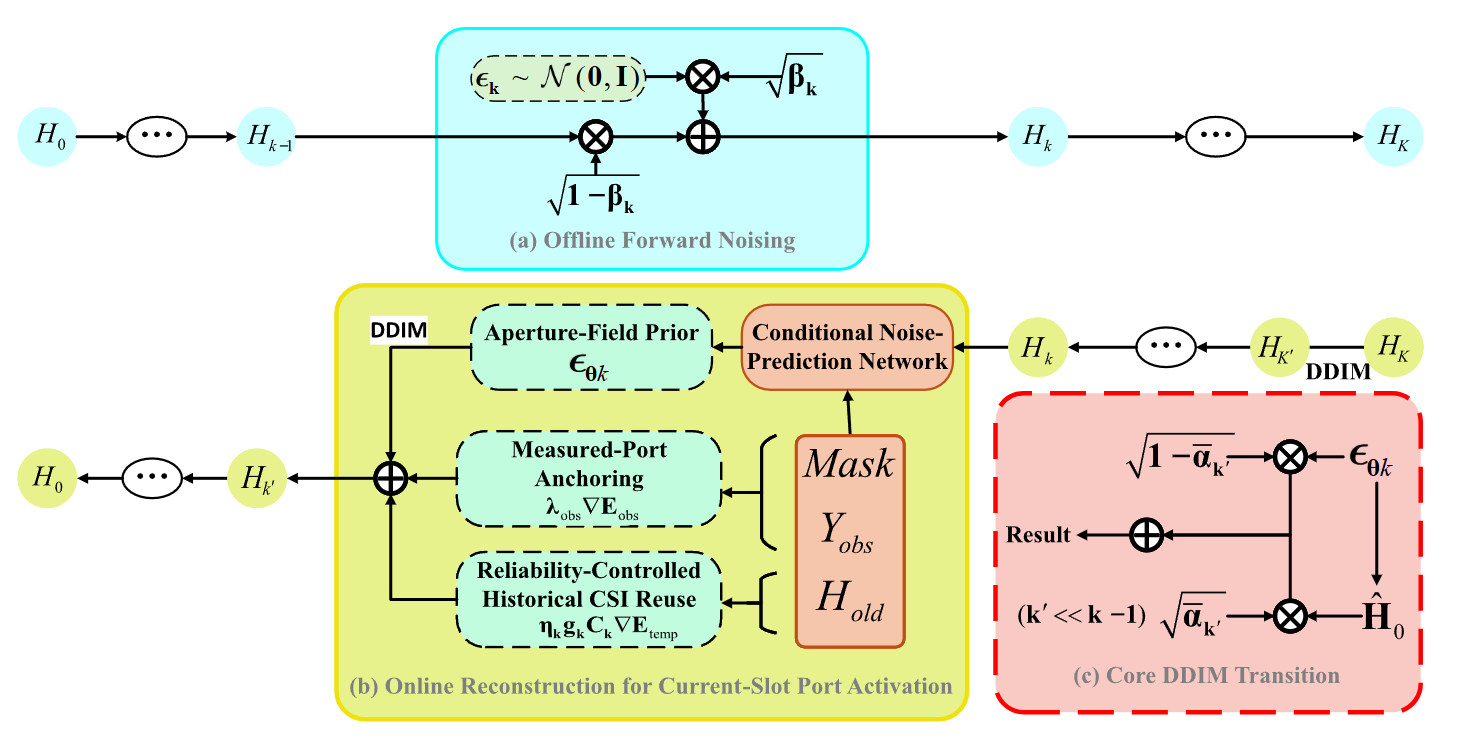}
    \caption{Overall framework of the proposed aging-aware conditional diffusion posterior reconstruction method. Offline, a conditional diffusion model learns the distribution of complete aperture fields. Online, the sparse current observations, the observation mask, and the historical CSI jointly guide posterior sampling so that the receiver can recover an aperture field suitable for current-slot port activation under limited probing.}
    \label{fig:framework}
\end{figure*}

\begin{figure}[t]
    \centering
    \includegraphics[width=\columnwidth]{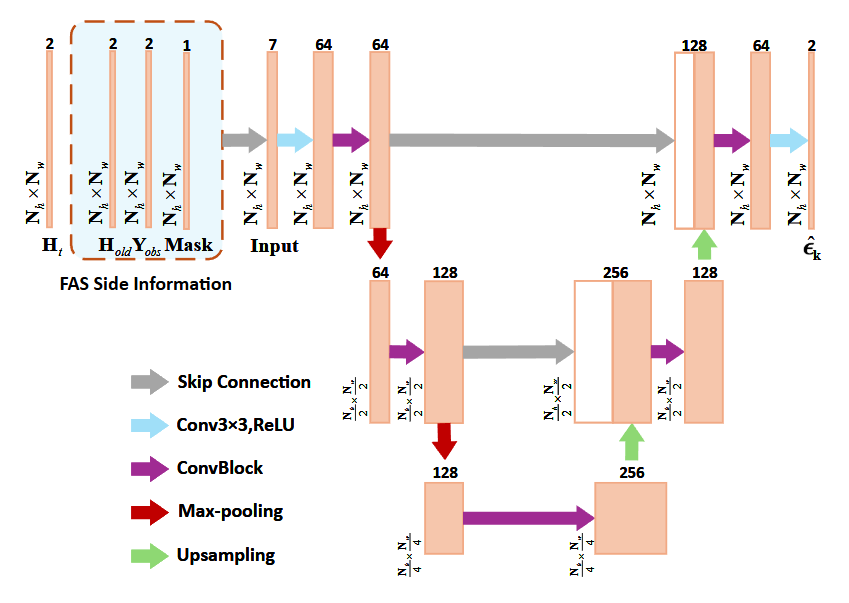}
    \caption{Conditional diffusion network architecture used for noise prediction. The network adopts a U-Net-style encoder--decoder with time-step embeddings and multi-scale feature propagation to preserve both global aperture structure and local field detail during denoising, which is important for maintaining decision-relevant port ranking in aperture-field reconstruction.}
    \label{fig:network}
\end{figure}

\subsection{Measured-Port Anchoring}
At inference time, the denoiser produces a clean-field estimate
\begin{equation}
\widehat{\mathbf{H}}^{(k)}_{t,0}=\mathcal{D}_\theta(\mathbf{H}_{t,k},k,\mathbf{C}_t).
\label{eq:cleanpred}
\end{equation}
This estimate is informed by the conditioning variables, but it is not guaranteed to match the measured entries exactly. In an FAS receiver, such drift is undesirable since the few ports probed in the current slot are the freshest direct evidence about the current aperture field.
Thus, we use the observation likelihood in \eqref{eq:obs_likelihood} as an explicit guidance term. For numerical stability across different probing budgets, the guidance energy is normalized as
\begin{equation}
\mathcal{E}_{\mathrm{obs}}^{(k)} = \frac{1}{2\sigma_w^2\sqrt{\|\mathbf{M}_t\|_0}}\left\|\mathbf{M}_t\odot\left(\widehat{\mathbf{H}}^{(k)}_{t,0}-\mathbf{Y}_t\right)\right\|_F^2.
\label{eq:eobs}
\end{equation}
The corresponding correction is
\begin{equation}
\widetilde{\mathbf{H}}^{(k)}_{t,0}=\widehat{\mathbf{H}}^{(k)}_{t,0}-\lambda_{\mathrm{obs}}\nabla \mathcal{E}_{\mathrm{obs}}^{(k)}.
\label{eq:obsgrad}
\end{equation}
This likelihood-gradient step corresponds to the observation term in \eqref{eq:score_decomp}. Observed ports act as hard evidence, while the diffusion prior controls the unmeasured region of the aperture field.

In addition to the gradient correction, we directly fuse the observed entries:
\begin{equation}
\overline{\mathbf{H}}^{(k)}_{t,0}=(\mathbf{1}-\mathbf{M}_t)\odot\widetilde{\mathbf{H}}^{(k)}_{t,0}+\mathbf{M}_t\odot\big[(1-a_t)\widetilde{\mathbf{H}}^{(k)}_{t,0}+a_t\mathbf{Y}_t\big],
\label{eq:obsfusion}
\end{equation}
where $a_t$ is a fusion factor. The gradient correction and direct fusion jointly limit drift at the currently measured entries, while the learned prior reconstructs the remaining aperture.

\subsection{Reliability-Controlled Reuse of Historical CSI}
Historical CSI is useful only when consistent with the current slot; otherwise, stale peaks can bias port ranking. DCP therefore estimates temporal reliability from pilot disagreement and controls historical reuse through sample-wise gating, spatial trust weighting, and progressive scheduling.

We first quantify the mismatch between the historical CSI and the sparse current observations through
\begin{equation}
\eta_t^{\mathrm{m}} = \frac{\|\mathbf{M}_t\odot(\mathbf{H}^{\mathrm{old}}_t-\mathbf{Y}_t)\|_F^2}{\|\mathbf{M}_t\odot\mathbf{Y}_t\|_F^2+\varepsilon},
\label{eq:etamismatch}
\end{equation}
where $\varepsilon$ is a small stabilizer. This quantity measures the normalized disagreement between the historical CSI and the currently observed ports. The sample-wise temporal gate is then defined as
\begin{equation}
 g_t = g_{\min} + (g_{\max}-g_{\min})\exp\!\left(-\frac{\eta_t^{\mathrm{m}}}{\tau_g}\right).
\label{eq:gt}
\end{equation}
Small mismatch keeps $g_t$ near $g_{\max}$, while large mismatch drives it toward $g_{\min}$ and weakens temporal correction. Thus, the receiver reuses historical CSI aggressively only when the currently observed ports indicate that the old aperture pattern still resembles the present one.

A scalar gate alone cannot distinguish spatial regions that remain reliable from those that have become obsolete. Hence, we construct a local discrepancy map
\begin{equation}
\mathbf{D}_t = \mathcal{B}\left(\mathbf{M}_t\odot\left|\mathbf{H}^{\mathrm{old}}_t-\mathbf{Y}_t\right|^2\right),
\label{eq:Dt}
\end{equation}
where $\mathcal{B}(\cdot)$ is a local smoothing operator that propagates observed disagreement to neighboring aperture locations. The spatial trust map is
\begin{equation}
\mathbf{T}_t = T_{\min} + (1-T_{\min})\exp\!\left(-\frac{\mathbf{D}_t}{\tau_T}\right).
\label{eq:Tt}
\end{equation}
Regions consistent with the sparse current observations receive larger temporal weights, while regions with larger historical-current mismatch receive smaller temporal weights.
This spatial selectivity is important as aging may shift or distort only part of the aperture field rather than the entire field uniformly.

The temporal correction is also scheduled progressively across the reverse diffusion trajectory. Let $r_k$ denote the reverse-sampling progress ratio. We activate temporal regularization only after a threshold $r_0$:
\begin{equation}
\lambda_{\mathrm{tmp}}^{(k)} =
\begin{cases}
0, & r_k<r_0,\\
\lambda_{\mathrm{tmp}}, & r_k\ge r_0.
\end{cases}
\label{eq:lambdatmp}
\end{equation}
This schedule lets early steps form a plausible current-slot field from the diffusion prior and the sparse current observations before historical regularization is activated. It reduces the chance that stale CSI dominates the reverse trajectory before the observed ports have established the present-slot structure.

Combining the temporal gate, the spatial trust map, and the progressive schedule, the historical-CSI-aware clean estimate is updated as
\begin{equation}
\mathbf{H}_{t,0}^{\star (k)} = \overline{\mathbf{H}}^{(k)}_{t,0} - \lambda_{\mathrm{tmp}}^{(k)} g_t (\mathbf{1}-\mathbf{M}_t)\odot \mathbf{T}_t \odot \left(\overline{\mathbf{H}}^{(k)}_{t,0}-\mathbf{H}^{\mathrm{old}}_t\right).
\label{eq:starclean}
\end{equation}
The correction applied only on unobserved ports; observed ports remain anchored by sparse current observations.   

\subsection{Low-Latency Sampling for Online FAS Receivers}
When the historical CSI is reliable, starting from pure Gaussian noise does not exploit the
available temporal structure. Hence, we use a gate-controlled warm initialization
\begin{equation}
\xi_t = \xi_{\min}+(\xi_{\max}-\xi_{\min})g_t,
\label{eq:xi}
\end{equation}
followed by
\begin{equation}
\mathbf{H}^{\mathrm{init}}_{t,K} = \xi_t\left(\sqrt{\bar{\alpha}_K}\,\mathbf{H}^{\mathrm{old}}_t+\sqrt{1-\bar{\alpha}_K}\,\mathbf{z}_1\right)+(1-\xi_t)\mathbf{z}_2,
\label{eq:warmstart}
\end{equation}
where $\mathbf{z}_1$ and $\mathbf{z}_2$ are Gaussian noise samples. A large $g_t$ initializes the sampler near a noisy historical field; a small $g_t$ recovers standard random initialization. For an online FAS receiver, this means that reliable historical CSI can provide a more informative initialization for obtaining a decision-ready aperture estimate along the fixed reverse trajectory.

For efficient inference, we use DDIM-based skipped sampling \cite{Song2020}. Let $\mathcal{K}=\{k_1,\dots,k_S\}$ denote the selected reverse steps with $k_1>\cdots>k_S=0$. The update can be summarized as
\begin{equation}
\mathbf{H}_{t,k_{s+1}} = \Phi_{\mathrm{DDIM}}\big(\mathbf{H}_{t,k_s},\mathbf{H}_{t,0}^{\star(k_s)},k_s,k_{s+1}\big),
\label{eq:ddimupdate}
\end{equation}
where $\Phi_{\mathrm{DDIM}}(\cdot)$ is the deterministic DDIM operator. Compared with a full-step stochastic reverse trajectory, DDIM sampling provides a favorable runtime--accuracy tradeoff for latency-sensitive online aperture-field reconstruction,  where the inference delay must remain small relative to the channel
coherence time and the probing overhead reduced by sparse
acquisition.

In general, averaging $L$ posterior samples yields the final estimate
\begin{equation}
\widehat{\mathbf{H}}_t = \frac{1}{L}\sum_{\ell=1}^{L}\mathbf{H}_{t,0}^{(\ell)}.
\label{eq:sampleavg}
\end{equation}
Increasing $L$ reduces sampling variance but increases latency. We set $L=1$ and use $25$ DDIM steps in the reported experiments. This choice reflects the intended operating point of an FAS receiver: produce one reliable current-slot aperture estimate quickly enough to support port activation, rather than spending extra latency to average many posterior draws. The dominant online cost is $\mathcal{O}(S\mathcal{C}_{\mathrm{net}})$, where $S$ is the number of retained DDIM steps and $\mathcal{C}_{\mathrm{net}}$ is the cost of one denoiser evaluation.

\begin{algorithm}[t]
\caption{Reliability-Controlled Reconstruction for Aging FAS Acquisition}
\label{alg:main}
\begin{algorithmic}[1]
\STATE \textbf{Input:} current observation $\mathbf{Y}_t$, observation mask $\mathbf{M}_t$, historical CSI $\mathbf{H}^{\mathrm{old}}_t$, condition tensor $\mathbf{C}_t$, reverse step set $\mathcal{K}$
\STATE Compare currently observed ports with historical CSI to compute mismatch metric $\eta_t^{\mathrm{m}}$ via \eqref{eq:etamismatch}, sample-wise temporal gate $g_t$ via \eqref{eq:gt}, and spatial trust map $\mathbf{T}_t$ via \eqref{eq:Tt}
\STATE Generate warm-start coefficient $\xi_t$ by \eqref{eq:xi} and initialize $\mathbf{H}_{t,K}^{\mathrm{init}}$ by \eqref{eq:warmstart}
\FOR{$s=1,2,\dots,S-1$}
    \STATE Predict the clean aperture field $\widehat{\mathbf{H}}_{t,0}^{(k_s)}$ with the conditional denoiser
    \STATE Anchor the currently observed ports using the observation-guidance correction in \eqref{eq:obsgrad}
    \STATE Fuse the current observations at the probed ports using \eqref{eq:obsfusion}
    \STATE Reuse historical CSI on unobserved ports through the reliability-controlled correction in \eqref{eq:starclean}
    \STATE Propagate to the next DDIM state through \eqref{eq:ddimupdate}
\ENDFOR
\STATE If multiple posterior samples are generated, average them using \eqref{eq:sampleavg}
\STATE \textbf{Output:} reconstructed full aperture field $\widehat{\mathbf{H}}_t$
\end{algorithmic}
\end{algorithm}

\section{Theoretical Validation of DCP}\label{sec:theory}
This section supports the proposed guidance design without claiming global optimality. We show that it implements meaningful per-step fusion, contracts measured-port residuals, and suppresses stale-CSI bias on unobserved entries.

\subsection{Per-Step Reliability-Calibrated Fusion}
To isolate the effect of the online guidance stage, fix a reverse step $k$ and freeze the denoiser prediction $\widehat{\mathbf{H}}_{t,0}^{(k)}$, the temporal gate $g_t$, and the trust map $\mathbf{T}_t$. Consider the surrogate objective
\begin{align}
\mathcal{Q}^{(k)}(\mathbf{X})
&=\frac{1}{2}\left\|\mathbf{X}-\widehat{\mathbf{H}}_{t,0}^{(k)}\right\|_F^2 \notag\\
&\quad+\frac{\alpha_t}{2}\left\|\mathbf{M}_t\odot(\mathbf{X}-\mathbf{Y}_t)\right\|_F^2 \notag\\
&\quad+\frac{\beta_t^{(k)}}{2}\left\|(\mathbf{1}-\mathbf{M}_t)\odot\mathbf{T}_t^{1/2}\odot(\mathbf{X}-\mathbf{H}^{\mathrm{old}}_t)\right\|_F^2,
\label{eq:Qk}
\end{align}
where
\begin{equation}
\alpha_t=\frac{\lambda_{\mathrm{obs}}}{\sigma_w^2\sqrt{\|\mathbf{M}_t\|_0}},
\qquad
\beta_t^{(k)}=\lambda_{\mathrm{tmp}}^{(k)}g_t .
\label{eq:alphabeta}
\end{equation}
The first term keeps the estimate close to the denoiser output, the second enforces current-slot observed ports, and the third reuses historical CSI only on unobserved ports.

\begin{proposition}[Closed-form guided fusion at each reverse step]\label{prop:closed_form}
The objective in \eqref{eq:Qk} is strictly convex and admits the unique minimizer
\begin{equation}
\mathbf{X}_{t}^{\dagger(k)}
=
\frac{
\widehat{\mathbf{H}}_{t,0}^{(k)}
\,+\,\alpha_t\mathbf{M}_t\odot\mathbf{Y}_t
\,+\,\beta_t^{(k)}(\mathbf{1}-\mathbf{M}_t)\odot\mathbf{T}_t\odot\mathbf{H}^{\mathrm{old}}_t
}{
\mathbf{1}
\,+\,\alpha_t\mathbf{M}_t
\,+\,\beta_t^{(k)}(\mathbf{1}-\mathbf{M}_t)\odot\mathbf{T}_t
},
\label{eq:closed_form}
\end{equation}
where the division is elementwise. Consequently, for each entry $(i,j)$,
\begin{align}
M_t(i,j)=1&\Rightarrow
X_{t,i,j}^{\dagger(k)}
=
\frac{\widehat{H}_{t,0,i,j}^{(k)}+\alpha_tY_{t,i,j}}{1+\alpha_t},
\label{eq:measured_entry}
\\
M_t(i,j)=0&\Rightarrow
X_{t,i,j}^{\dagger(k)}
=
\frac{\widehat{H}_{t,0,i,j}^{(k)}+\beta_t^{(k)}T_{t,i,j}H^{\mathrm{old}}_{t,i,j}}{1+\beta_t^{(k)}T_{t,i,j}}.
\label{eq:unmeasured_entry}
\end{align}
\end{proposition}
\noindent\textit{Proof:}
The objective \eqref{eq:Qk} is a sum of strictly convex quadratic terms, so it has a unique minimizer. Since the Frobenius norm separates entrywise, each entry can be optimized independently. Setting the derivative of \eqref{eq:Qk} with respect to each entry to zero gives
\begin{align}
\mathbf{0}
&=
\mathbf{X}-\widehat{\mathbf{H}}_{t,0}^{(k)}
\,+\,\alpha_t\mathbf{M}_t\odot(\mathbf{X}-\mathbf{Y}_t) \notag\\
&\quad+\beta_t^{(k)}(\mathbf{1}-\mathbf{M}_t)\odot\mathbf{T}_t\odot(\mathbf{X}-\mathbf{H}^{\mathrm{old}}_t),
\end{align}
which can be rearranged entrywise to obtain \eqref{eq:closed_form}. The measured-port and unmeasured-port forms in \eqref{eq:measured_entry}--\eqref{eq:unmeasured_entry} follow directly by substituting $M_t(i,j)=1$ and $M_t(i,j)=0$, respectively. $\hfill\blacksquare$

Proposition~\ref{prop:closed_form} gives a direct interpretation of the proposed guidance structure. 
The estimates at observed ports are obtained by fusing the denoiser prediction with the sparse current observations.
Unobserved ports are updated by a reliability-weighted interpolation between the denoiser prediction and the historical CSI. The sequential guidance steps in Section~\ref{sec:method} provide a low-cost approximation to this reliability-calibrated fusion rule within the reverse sampler.

\subsection{Measured-Port Stability and Stale-CSI Bias Control}
The proposed sampler is not a convex optimization algorithm, and we therefore do not claim global convergence. Instead, we analyze two monotonicity properties that are directly tied to the FAS acquisition mechanism: contraction of the residual on currently observed ports and suppression of bias induced by stale historical CSI.

\begin{proposition}[Observation-residual contraction]\label{prop:obs_contract}
Let $\mathbf{R}=\mathbf{M}_t\odot(\mathbf{X}-\mathbf{Y}_t)$ be the residual on the observed ports. Consider the observation-guidance update $\mathbf{X}^+=\mathbf{X}-\mu\mathbf{M}_t\odot(\mathbf{X}-\mathbf{Y}_t)$ with $0<\mu<2$. Then
\begin{equation}
\left\|\mathbf{M}_t\odot(\mathbf{X}^+-\mathbf{Y}_t)\right\|_F
=|1-\mu|\left\|\mathbf{R}\right\|_F
<\left\|\mathbf{R}\right\|_F.
\label{eq:obs_contract}
\end{equation}
\end{proposition}
\noindent\textit{Proof:} Since $\mathbf{M}_t\odot\mathbf{M}_t=\mathbf{M}_t$, applying the mask to the update gives $\mathbf{M}_t\odot(\mathbf{X}^+-\mathbf{Y}_t)=(1-\mu)\mathbf{M}_t\odot(\mathbf{X}-\mathbf{Y}_t)$. Taking the Frobenius norm yields \eqref{eq:obs_contract}. $\hfill\blacksquare$

Proposition~\ref{prop:obs_contract} explains why \eqref{eq:obsgrad} and \eqref{eq:obsfusion} reduce drift on measured entries. The direct fusion step strengthens this effect by explicitly replacing or averaging the measured components with the sparse current observations. In the FAS context, this ensures that the few ports actually sounded in the current slot remain fixed reference points throughout reconstruction.

\begin{proposition}[Bias control under stale CSI]\label{prop:bias_control}
Let $\bm{\Delta}_t=\mathbf{H}^{\mathrm{old}}_t-\mathbf{H}_t$ denote the historical CSI mismatch. For the temporal correction in \eqref{eq:starclean}, the bias magnitude induced by stale CSI on unobserved entries satisfies
\begin{align}
&\left\|\lambda_{\mathrm{tmp}}^{(k)}g_t(\mathbf{1}-\mathbf{M}_t)
\odot\mathbf{T}_t\odot\bm{\Delta}_t\right\|_F^2 \notag\\
&\quad\le
(\lambda_{\mathrm{tmp}}^{(k)})^2 g_t^2 \|\mathbf{T}_t\|_\infty^2
\left\|(\mathbf{1}-\mathbf{M}_t)\odot\bm{\Delta}_t\right\|_F^2.
\label{eq:bias_bound}
\end{align}
Moreover, since $g_t=g_{\min}+(g_{\max}-g_{\min})\exp(-\eta_t^{\mathrm{m}}/\tau_g)$, the upper bound decreases monotonically with the observed historical mismatch $\eta_t^{\mathrm{m}}$ whenever $g_{\min}$ is small.
\end{proposition}
\noindent\textit{Proof:} The result follows from the submultiplicative property of the elementwise weighted Frobenius norm and $0\le \mathbf{T}_t\le \|\mathbf{T}_t\|_\infty$. The monotonicity follows from $\partial g_t/\partial \eta_t^{\mathrm{m}}=-(g_{\max}-g_{\min})\exp(-\eta_t^{\mathrm{m}}/\tau_g)/\tau_g<0$. $\hfill\blacksquare$

Proposition~\ref{prop:bias_control} formalizes the intended role of reliability-aware temporal reuse. If the historical CSI is consistent with the sparse current observations, $g_t$ is large and the posterior sampler can exploit temporal side information. If the historical CSI is stale, the induced bias is suppressed by the sample-wise gate and the spatial trust map. In practical terms, the mechanism reduces the risk that outdated aperture peaks distort the reconstructed port ranking. This explains why a fixed temporal weight is less robust than the proposed adaptive weighting.

\subsection{Computational Implication}
From a complexity perspective, let $S$ denote the number of retained DDIM steps, $L$ the number of posterior samples, and $\mathcal{C}_{\mathrm{net}}$ the cost of one denoiser evaluation. The dominant online complexity is $\mathcal{O}(LS\mathcal{C}_{\mathrm{net}})$, while observation fusion, mismatch evaluation, trust-map construction, and temporal correction require only lower-order elementwise operations of $\mathcal{O}(N_hN_w)$. Hence, with the reported setting $L=1$ and $S=25$, DCP is dominated by 25 denoiser evaluations, and its reliability-control modules do not change the asymptotic order. Memory is likewise dominated by the U-Net parameters and intermediate feature maps. Because DCP, DUPP, and DUNP use the same number of retained denoising steps and closely related backbones, their deployment costs are expected to be similar, with only a small additional cost for reliability control. Section~\ref{sec:results} verifies this expectation using measured FLOP and peak-memory statistics.

\section{Dataset Construction and Experimental Setup}\label{sec:setup}
\subsection{Data Construction and Aging Protocol}
All methods use a common synthetic $16\times16$ ($256$-port) aperture pipeline. Spatial fields follow \eqref{eq:spatialcorr}, temporal evolution follows \eqref{eq:aging}, and historical CSI is constructed from $\mathbf{H}^{\mathrm{c}}_{t-d}$ through \eqref{eq:hold}, where $d$ sets the source slot and $\alpha_t$, $\mathbf{S}^{\mathrm{c}}_t$, and $\mathbf{N}^{\mathrm{old,c}}_t$ model additional degradation. Sparse observations follow \eqref{eq:obsmodel}. This controlled protocol isolates sparse probing, temporal correlation, and historical-CSI mismatch rather than replacing ray-tracing or measured-channel campaigns.

Unless otherwise specified, the aperture is $16\times16$ with $256$ ports, the SNR is $15$ dB, and diffusion sampling uses $25$ DDIM steps. Budget, temporal-correlation, and observation-pattern settings are summarized in Table~\ref{tab:sim}; $\rho_t$ is reported directly to avoid tying the results to a specific carrier, velocity, or slot duration.

\subsection{Compared Methods and Fairness Rules}
The compared methods are DCP, DUPP, DUNP, orthogonal
matching pursuit (OMP), sparse Bayesian learning (SBL), linear minimum mean-square error (LMMSE), and spatiotemporal Gaussian-process estimator (GP-ST). DCP is the full proposed design. DUPP keeps the online posterior-prior correction but removes conditional diffusion training; it therefore tests whether historical CSI can be added only at inference time. DUNP removes historical CSI reuse and keeps an unconditional diffusion prior; it tests whether a generative aperture prior alone is sufficient. OMP and SBL are model-driven sparse-recovery baselines. LMMSE is the standard covariance-based linear estimator. GP-ST is the most relevant analytical benchmark since it fuses sparse current observations and temporal side information through a Gaussian process prior.

All methods share the same data partitions, masks, and training-derived normalization, preventing test leakage and sampling-set advantages. In the observation-pattern tests, all methods use identical observed ports within each pattern; latency is measured after warm-up in the same hardware environment and reported per sample. These controls ensure that the reported differences primarily reflect the reconstruction strategy and its use of current or historical CSI rather than differences in data partitioning, probing support, or timing conditions.

\begin{table}[t]
\caption{Core Experimental Configuration}
\label{tab:sim}
\centering
\small
\begin{tabular}{p{2.75cm}p{3.75cm}}
\toprule
Item & Setting \\
\midrule
Aperture grid & $16\times16$ \\
Total ports & $256$ \\
Training epochs & $250$ \\
Default SNR & $15$ dB \\
Main DDIM steps & $25$ \\
Budget sweep & $20,40,60,80,100,120,140$ ports \\
Temporal-correlation bins & $[0.75,0.80]$ to $[0.95,1.00]$ \\
Observation patterns & Uniform, Random, Greedy \\
Pattern-robustness budget & $20$ ports \\
Latency metric & per-sample online inference latency \\
\bottomrule
\end{tabular}
\end{table}

\begin{table}[t]
\caption{Diffusion Baselines Used in the Comparison}
\label{tab:ablation}
\centering
\small
\begin{tabular}{p{1.3cm}p{2.0cm}p{1.75cm}p{2.1cm}}
\toprule
Method & Offline input & Online temporal path & Purpose \\
\midrule
DCP & $\mathbf{H}^{\mathrm{old}}_t$, $\mathbf{Y}_t$, $\mathbf{M}_t$ & Soft historical regularization retained & Full conditional design \\
DUPP & Aperture-field samples only & Use $\mathbf{H}^{\mathrm{old}}_t$ only in posterior correction & Isolate online historical CSI injection \\
DUNP & Aperture-field samples only & No historical prior & Diffusion baseline without historical CSI \\
\bottomrule
\end{tabular}
\end{table}

\subsection{Evaluation Metrics}
The primary plotted reconstruction metric is normalized mean-square error (NMSE),
\begin{equation}
\mathrm{NMSE}_t = \frac{\|\widehat{\mathbf{H}}_t-\mathbf{H}_t\|_F^2}{\|\mathbf{H}_t\|_F^2},
\label{eq:nmse}
\end{equation}
which is also plotted in decibel form as $10\log_{10}(\mathrm{NMSE}_t)$ for clearer visualization of performance trends. The online-cost metric is the average per-sample inference time under a fixed hardware platform after warm-up. For communication-level evaluation, we additionally report the average achievable rate of the reconstruction-based selected port according to \eqref{eq:rate_def}, using a transmission SNR of $5$ dB. NMSE measures full-field fidelity, whereas achievable rate directly reflects whether the reconstruction preserves decision-relevant high-gain ports. Together with latency, these metrics separate reconstruction accuracy, communication-level port-selection utility, and online deployment cost.

\subsection{Parameter Interpretation}
The inference parameters have direct operational roles and can be tuned hierarchically rather than through an exhaustive joint search. The observation-guidance weight $\lambda_{\mathrm{obs}}$ controls measured-port anchoring and should increase with current-measurement reliability; smaller values are preferable when the sparse observations are noisy to avoid overfitting unreliable pilots. The range $[g_{\min},g_{\max}]$ bounds historical-CSI reuse, while $\tau_g$ controls its sensitivity to the normalized mismatch in \eqref{eq:etamismatch}; $g_{\min}$ should remain small enough to suppress severely stale CSI and $g_{\max}$ should permit strong reuse when old and current measurements agree. Similarly, $T_{\min}$ and $\tau_T$ set the minimum local trust and its sensitivity to the discrepancy map in \eqref{eq:Dt}. The threshold $r_0$ is preferably selected in the middle-to-late reverse trajectory so that current-slot structure is established before temporal regularization becomes influential, while $[\xi_{\min},\xi_{\max}]$ should provide a moderate warm start that exploits useful history without biasing the initialization excessively toward stale CSI. For a new deployment, we recommend tuning $\lambda_{\mathrm{obs}}$ over representative SNRs, then the gate/trust parameters over representative aging levels, and finally $(r_0,\xi_{\min},\xi_{\max})$ on a held-out validation set; the selected values should then be kept fixed during the corresponding test sweeps.

\section{Numerical Results}\label{sec:results}
The experiments evaluate the proposed framework in terms of reconstruction accuracy, selected-port achievable rate, robustness, and online complexity.
We first examine the benefit of conditioning the diffusion prior on sparse current observations and stale CSI for instantaneous aperture-field reconstruction. 
We then quantify how the reconstruction performance varies with the probing budget, SNRs, temporal correlation, historical delay, and
historical-prior noise, and finally assess the associated latency overhead.
The discussion focuses on the main performance trends, their underlying physical interpretations, and the resulting design implications.

\subsection{Main Reconstruction Results}
The first group of results varies the sparse current observations, the quality of that evidence, and the reliability of the historical CSI. These tests correspond to the main operating factors of an aging FAS receiver.

\subsubsection{Effect of Probing Budget}
Figs.~\ref{fig:budget_results} and~\ref{fig:achievable_rate_results} jointly characterize the probing-budget tradeoff. As the number of observed ports increases from $20$ to $140$, all methods improve, while DCP maintains the lowest NMSE and achieves the highest average achievable rate, with the clearest advantage in the sparse-probing regime. This gain follows from using historical CSI during both training and sampling: DUNP lacks current/temporal conditioning, DUPP injects historical CSI only online, and GP-ST is constrained by a fixed correlation model. The simultaneous NMSE and rate gains indicate that DCP improves not only global aperture reconstruction but also the decision-relevant structure used for port activation. At larger budgets, SBL approaches DCP in achievable rate because rate mainly depends on preserving a near-optimal high-gain port rather than minimizing full-field NMSE, consistent with Proposition~\ref{prop:rate_regret}.

\begin{figure}[!t]
    \centering
    \includegraphics[width=\columnwidth]{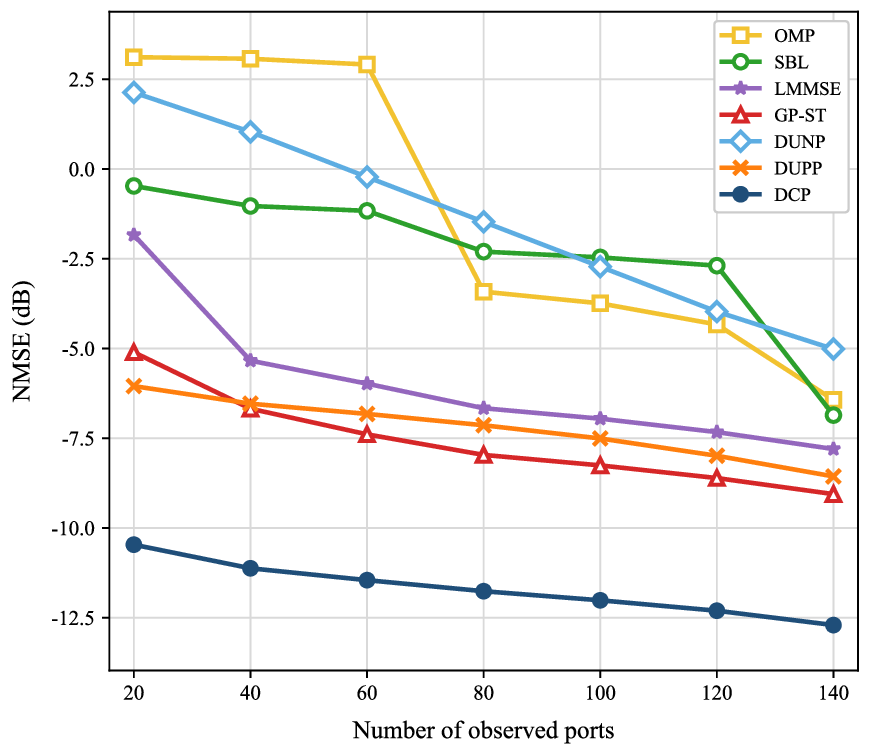}
    \caption{NMSE versus the number of observed ports, which varies from 20 to 140.}
    \label{fig:budget_results}
\end{figure}

\begin{figure}[!t]
    \centering
    \includegraphics[width=\columnwidth]{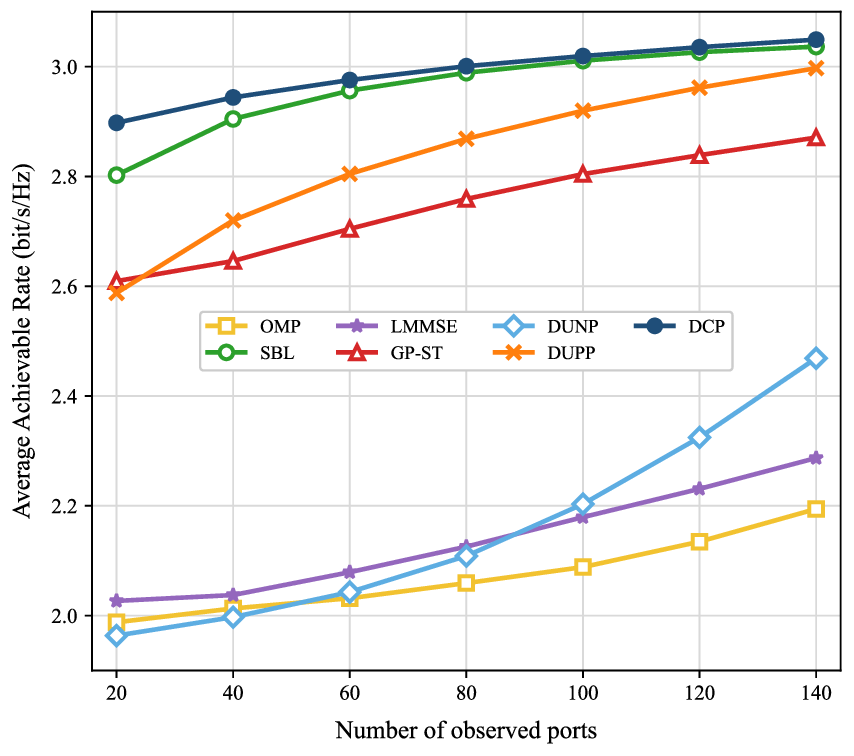}
    \caption{Average achievable rate versus the number of observed ports under random probing, with an observation SNR of $15$ dB and a transmission SNR of $5$ dB.}
    \label{fig:achievable_rate_results}
\end{figure}

\subsubsection{Effect of SNR}
Fig.~\ref{fig:snr_results} shows the SNR sweep from $-10$ to $30$ dB. The NMSE of all methods decreases as the sparse current observations become cleaner, and DCP remains best throughout the range. At low SNR, noisy pilots limit reconstruction, whereas at high SNR the dominant errors arise from unobserved ports and model mismatch; consequently, LMMSE and GP-ST flatten while DCP continues improving by propagating reliable local evidence through its learned aperture prior.

\begin{figure}[!t]
    \centering
    \includegraphics[width=\columnwidth]{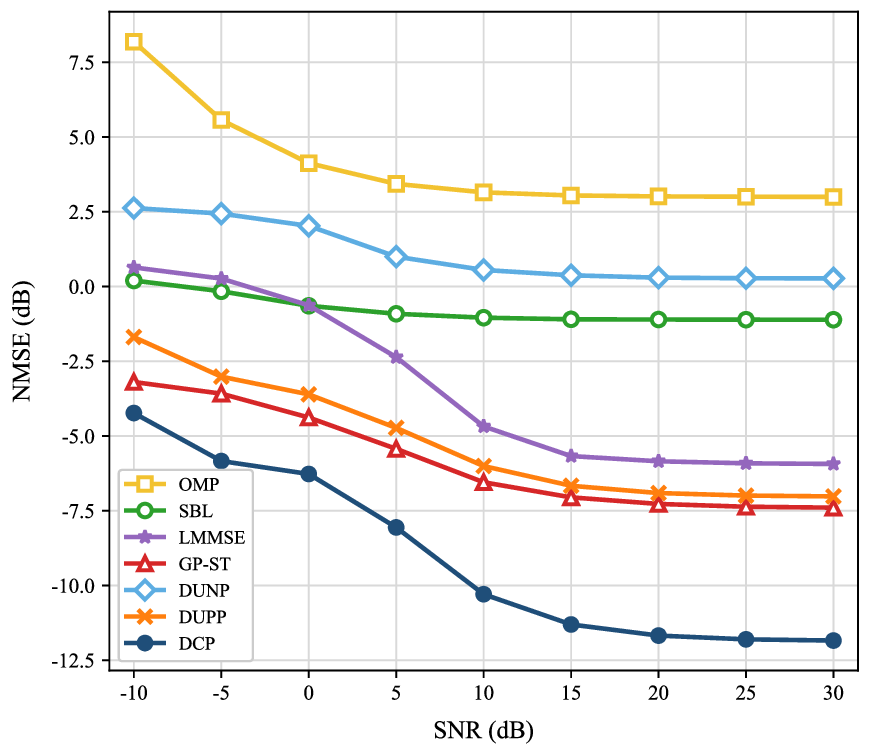}
    \caption{NMSE versus SNR over the range from $-10$ to $30$ dB at a probing ratio of $0.2$.}
    \label{fig:snr_results}
\end{figure}

\subsubsection{Effect of Temporal Correlation}
Fig.~\ref{fig:rho_results} shows that methods exploiting historical CSI improve as temporal correlation increases, while DUNP is comparatively insensitive because it ignores temporal information and GP-ST remains constrained by its analytical model. At low correlation, stale CSI is useful only where it agrees with the current observations and can otherwise introduce bias; at high correlation, it becomes a stronger predictor for unobserved ports. DCP handles both regimes by using the mismatch gate and spatial trust map to suppress unreliable historical CSI and strengthen temporal reuse when the side information is more reliable.

\begin{figure}[!t]
    \centering
    \includegraphics[width=\columnwidth]{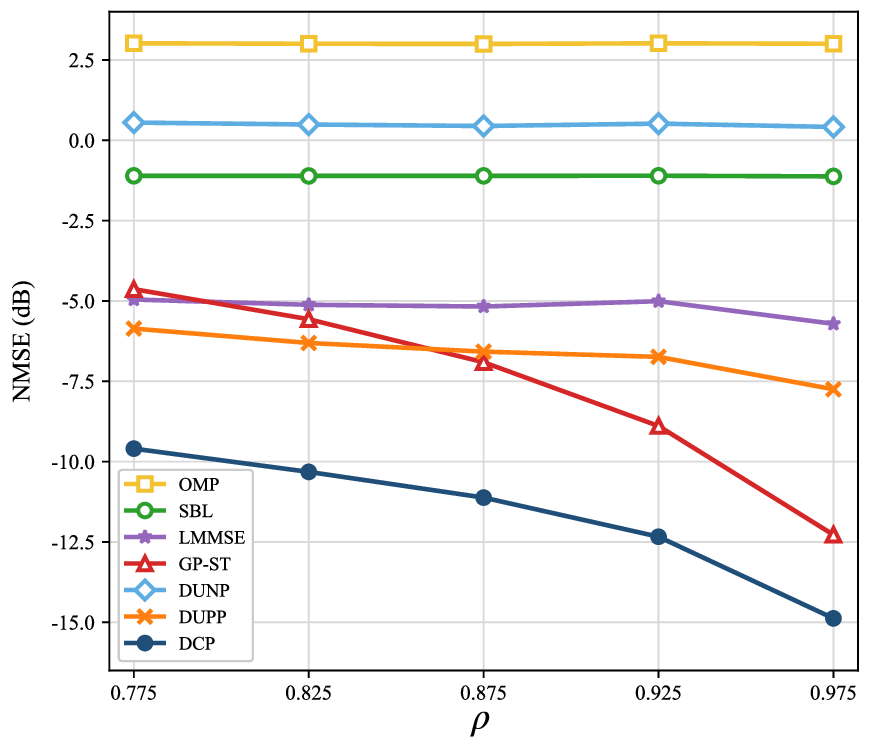}
    \caption{NMSE versus temporal correlation over five intervals spanning $\rho_t\in[0.75,1.00]$ at an SNR of $15$ dB.}
    \label{fig:rho_results}
\end{figure}

\subsection{Mechanism Validation}
The next experiments check whether the gain comes from the reliability-control mechanism itself and whether the method remains stable when the historical CSI is delayed or noisy and when the current observations follow different probing patterns.

\subsubsection{Effect of Historical Delay}
Fig.~\ref{fig:hist_delay_results} varies the historical delay $d$ using $\mathbf{H}^{\mathrm{old}}_t=\mathbf{H}_{t-d}$ without additional Doppler-dependent scaling, thereby isolating the effect of CSI age from the extra fidelity degradation modeled by $\alpha_t$ in \eqref{eq:hold}. All variants degrade as $d$ increases, and the hard-prior variant deteriorates most severely because outdated CSI is imposed too strongly. DCP-full remains best across the delay range; its advantage over the no-prior variant gradually narrows as the sampler assigns less trust to increasingly stale CSI. This behavior is consistent with Proposition~\ref{prop:aging_rate}: temporal decorrelation should progressively reduce, rather than strengthen, the role of historical CSI.

\begin{figure}[!t]
    \centering
    \includegraphics[width=\columnwidth]{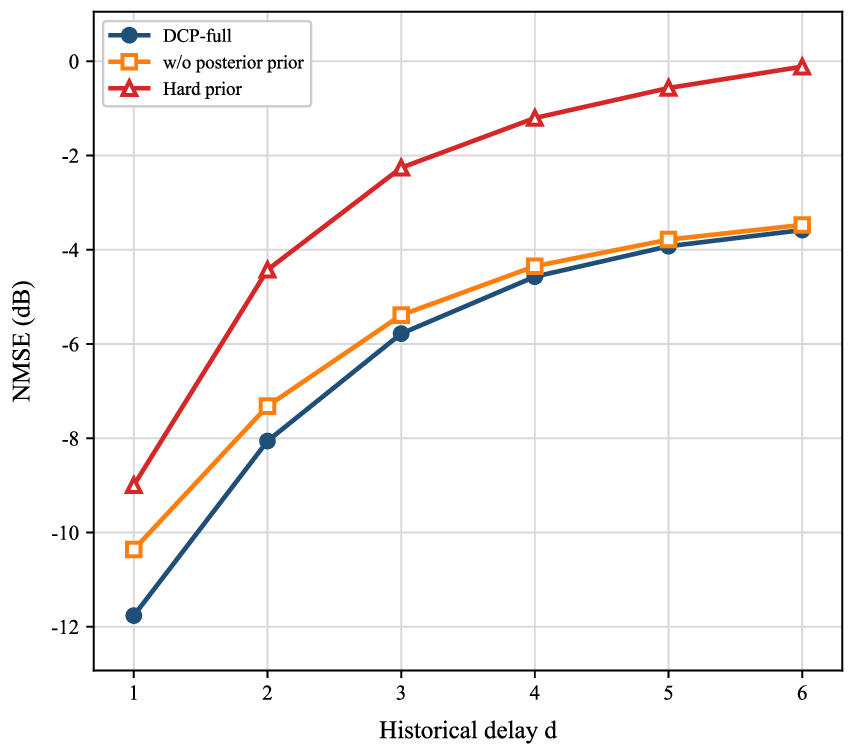}
    \caption{NMSE of the DCP variants versus historical delay
from $1$ to $6$ slots with 100 observed ports at a pilot SNR
of 15 dB.}
    \label{fig:hist_delay_results}
\end{figure}

\subsubsection{Effect of Historical-Prior Noise}
Fig.~\ref{fig:hist_noise_results} evaluates robustness to corrupted historical CSI by applying $\mathbf{H}^{\mathrm{old}}\leftarrow\mathbf{H}^{\mathrm{old}}+\sigma_{\mathrm{old}}\mathbf{Z}$. NMSE increases with $\sigma_{\mathrm{old}}$ for all variants, and the hard-prior design is the most sensitive, confirming that noisy historical CSI should not be treated as ground truth. DCP-full degrades more slowly, while removing posterior-prior injection, warm initialization, or reliability control becomes increasingly harmful under stronger corruption. The widening separation under larger $\sigma_{\mathrm{old}}$ shows that the gate and trust map are most useful when historical CSI still carries exploitable temporal structure but its local reliability is uncertain, which is precisely the mixed-reliability regime targeted by DCP.

\begin{figure}[!t]
    \centering
    \includegraphics[width=\columnwidth]{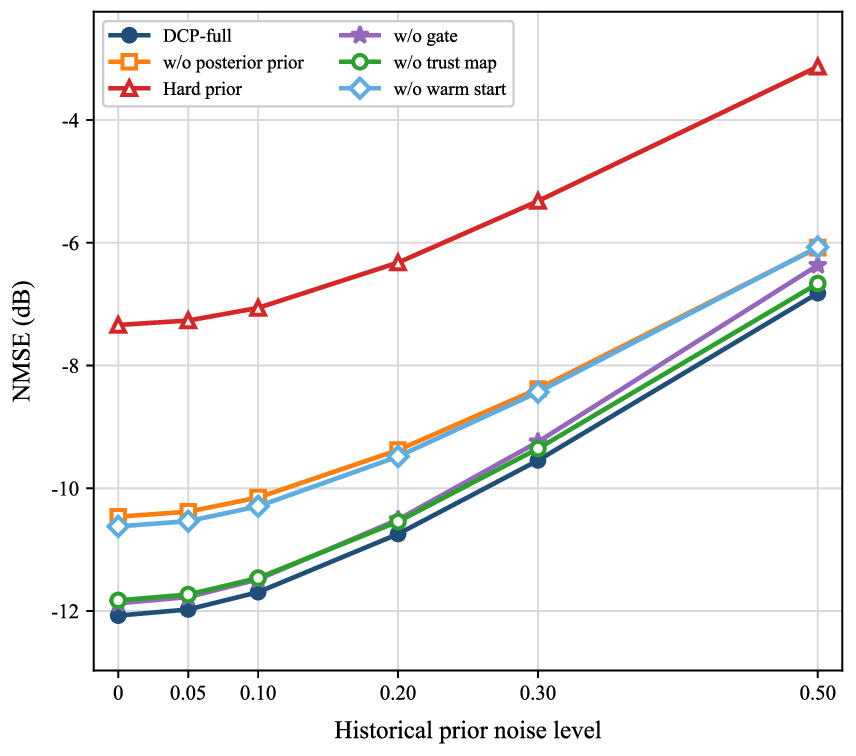}
    \caption{NMSE of DCP variants versus historical-prior noise level $\sigma_{\mathrm{old}}$ as $\sigma_{\mathrm{old}}$ varies from $0$ to $0.50$ with $100$ observed ports at an SNR of $15$ dB.}
    \label{fig:hist_noise_results}
\end{figure}

\subsubsection{Historical-Prior Ablation}
Fig.~\ref{fig:ablation_results} compares DCP with five reduced variants. DCP-full achieves the lowest NMSE at every tested budget, whereas the hard-prior variant performs worst, confirming that the gain comes from controlled historical reuse rather than simply adding stale CSI. The reduced variants separate the roles of posterior-prior injection, warm initialization, sample-wise gating, and spatial trust weighting: the first two improve the use of informative temporal structure, while the gate and trust map limit the influence of unreliable historical guidance. The contrast with the hard-prior design further shows that historical CSI is most useful when its strength and spatial support are adaptively calibrated rather than enforced uniformly.

\begin{figure}[!t]
    \centering
    \includegraphics[width=\columnwidth]{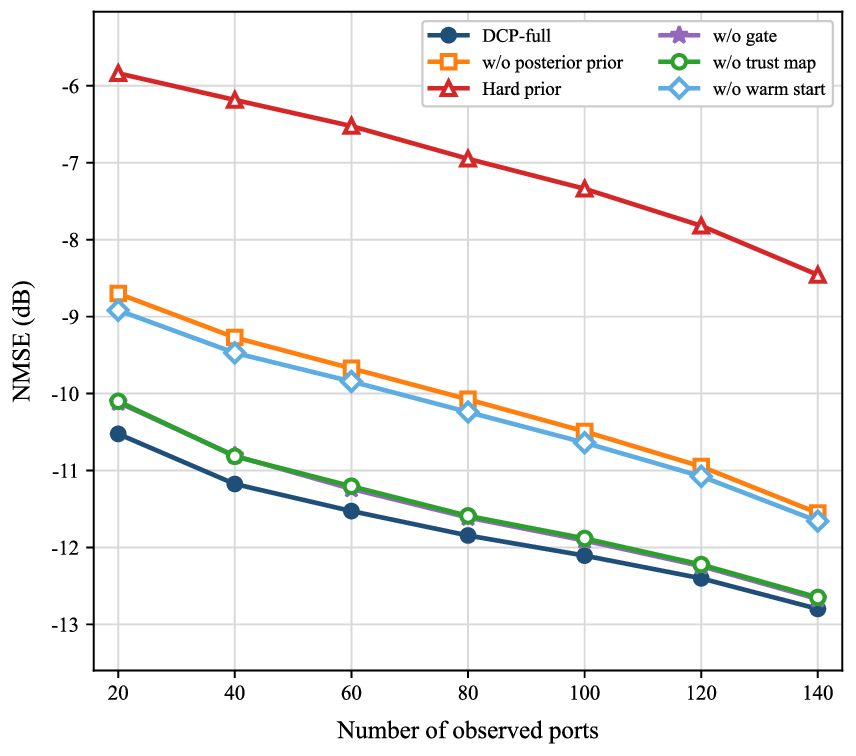}
    \caption{NMSE comparison of DCP and five ablated variants as the number of observed ports varies from $20$ to $140$ at an SNR of $15$ dB.}
    \label{fig:ablation_results}
\end{figure}

\subsubsection{Robustness to Observation Pattern}
Fig.~\ref{fig:pattern_results} compares Uniform, Random, and Greedy observation patterns under the same sparse budget. DCP exhibits only limited NMSE variation across the three patterns, indicating that its gain is not tied to one mask geometry. Because the observation mask is explicitly included in the condition tensor, the denoiser is informed of which ports are actually observed and the same reconstruction rule can adapt to different spatial supports without assuming a fixed probing layout. This robustness is important for practical FAS acquisition, where the available probing pattern may change with pilot budget or port-selection policy.

\begin{figure}[!t]
    \centering
    \includegraphics[width=\columnwidth]{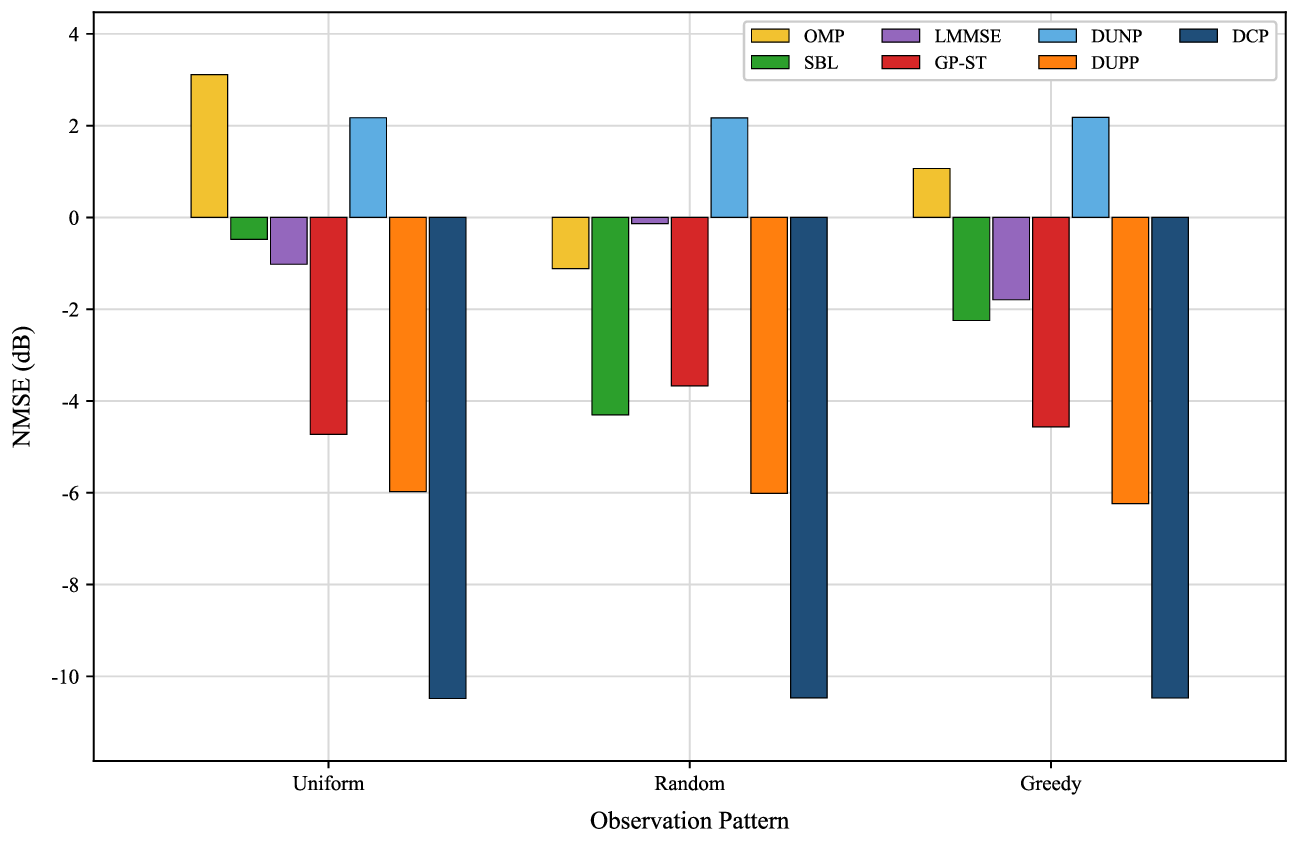}
    \caption{NMSE comparison across Uniform, Random, and Greedy observation patterns with a fixed budget of $20$ observed ports at an SNR of $15$ dB.}
    \label{fig:pattern_results}
\end{figure}

\subsection{Complexity and Online Deployment}
We quantify online cost using approximate FLOPs, trainable parameters, peak GPU memory, and measured inference latency. The FLOP/memory profile uses batch size $1$ and FP32 on an NVIDIA GeForce RTX 4090 at a probing ratio of $0.2$ and an SNR of $15$ dB, with $25$ DDIM steps and $L=1$ for the diffusion methods. Fig.~\ref{fig:method_latency_results} shows that GP-ST and LMMSE have the lowest latency, OMP/SBL require iterative recovery, and DCP has nearly the same few-millisecond runtime as the other diffusion variants. Diffusion FLOPs count the dominant convolutional/linear operations, while the OMP/SBL values are upper-bound estimates because early stopping can reduce the realized computation. Table~\ref{tab:complexity_profile} shows that DCP requires $5.463$ GFLOPs and $19.18$ MiB peak memory, while adding only $0.68\%$ FLOPs, $0.14\%$ parameters, and $0.08\%$ peak memory over DUPP. This measured profile is consistent with the computational analysis in Section~\ref{sec:theory}: repeated denoiser evaluations dominate the diffusion cost, whereas mismatch evaluation, trust-map construction, and temporal correction contribute only lower-order operations. Therefore, reliability calibration introduces negligible additional overhead relative to the diffusion backbone, while LMMSE and GP-ST remain substantially lighter analytical alternatives when computational cost is the primary concern.

\begin{figure}[!t]
    \centering
    \includegraphics[width=\columnwidth]{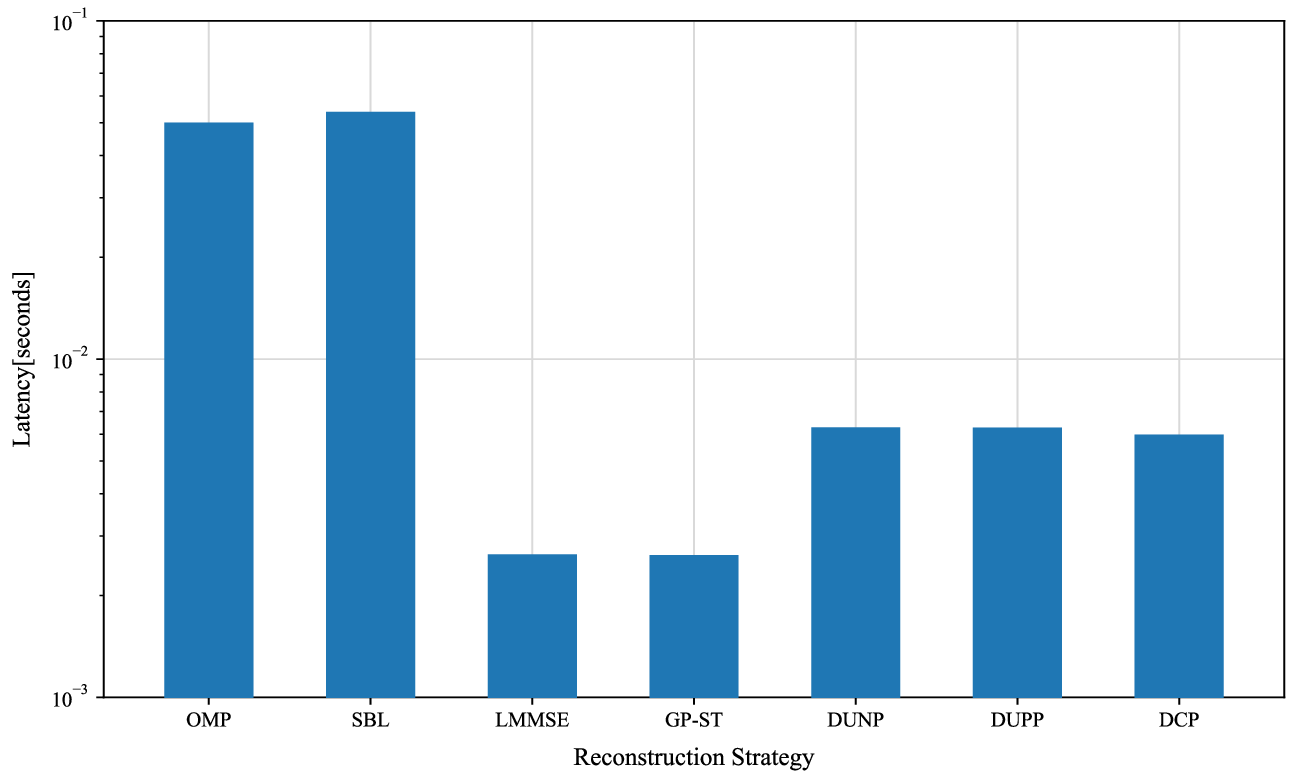}
    \caption{Average per-sample online inference latency across seven reconstruction strategies at a probing ratio of $0.2$ and an SNR of $15$ dB, shown on a logarithmic scale.}
    \label{fig:method_latency_results}
\end{figure}

\begin{table}[!t]
\caption{Computational Profile of the Compared Methods}
\label{tab:complexity_profile}
\centering
\small
\setlength{\tabcolsep}{3.2pt}
\begin{tabular}{lccc}
\toprule
Method & Approx. FLOPs (G) & Params (M) & Peak memory (MiB) \\
\midrule
OMP    & 0.0122 & --    & 14.76 \\
SBL    & 4.356  & --    & 10.87 \\
LMMSE  & $4.79\times10^{-4}$ & -- & 9.29 \\
GP-ST  & $4.79\times10^{-4}$ & -- & 8.79 \\
DUNP   & 5.426  & 2.017 & 19.15 \\
DUPP   & 5.426  & 2.017 & 19.17 \\
DCP    & 5.463  & 2.020 & 19.18 \\
\bottomrule
\end{tabular}
\end{table}

\section{Conclusion}\label{sec:conclusion}

\label{sec:conclusion}
This paper studied sparse acquisition in aging two-dimensional FAS receivers, where fresh-but-sparse observations must be combined with dense-but-stale CSI. DCP integrates conditional diffusion with observation anchoring, adaptive temporal gating, spatial trust weighting, progressive activation, and warm initialization, while the analysis links temporal decorrelation and reconstruction error to port-selection rate distortion and regret. The numerical results show improved reconstruction and selected-port achievable rate under limited probing and imperfect historical CSI, with negligible additional FLOP/memory cost over the diffusion baselines and millisecond-level inference. The complementary NMSE and rate results further indicate that the proposed reconstruction does not only improve full-field fidelity, but also better preserves decision-relevant high-gain ports for current-slot activation. Overall, the results support a practical design principle for aging FAS acquisition: historical CSI should neither be discarded nor enforced as a hard prior, but reused according to its current reliability. This reliability-calibrated temporal reuse reduces sparse-acquisition uncertainty while maintaining a favorable balance among probing overhead, communication utility, and online deployment cost. Future work will consider measured channels, multi-slot memory, adaptive probing, and direct communication-utility optimization.

\balance


\begin{thebibliography}{99}
	
\bibitem{Lai2026JSAC}
X. Lai {\em et al.}, ``Revisiting spatial block-correlation model for fluid antenna systems: From constant to variable correlations,'' \emph{IEEE J. Sel. Areas Commun.}, vol. 44, pp. 1335--1351,
Oct. 2025.

\bibitem{Shojaeifard2022ProcIEEE}
A. Shojaeifard {\em et al.}, ``MIMO evolution beyond 5G through reconfigurable intelligent surfaces and fluid antenna systems,'' \emph{Proc. IEEE}, vol. 110, no. 9, pp. 1244--1265, Sep. 2022.

\bibitem{Wong2021TWC}
K.-K. Wong, A. Shojaeifard, K.-F. Tong, and Y. Zhang, ``Fluid antenna systems,'' \emph{IEEE Trans. Wireless Commun.}, vol. 20, no. 3, pp. 1950--1962, Mar. 2021.

\bibitem{New2024Outage}
W. K. New, K.-K. Wong, H. Xu, K.-F. Tong, and C.-B. Chae, ``Fluid antenna system: New insights on outage probability and diversity gain,'' \emph{IEEE Trans. Wireless Commun.}, vol. 23, no. 1, pp. 128--140, Jan. 2024.


\bibitem{New2025Tutorial}
W. K. New {\em et al.}, ``A tutorial on fluid antenna system for 6G networks: Encompassing communication theory, optimization methods and hardware designs,'' \emph{IEEE Commun. Surveys Tuts.}, vol. 27, no. 4, pp. 2325--2377, Aug. 2025.

\bibitem{Wu2025MWC}
T. Wu {\em et al.}, ``Fluid antenna systems enabling 6G: Principles, applications, and research directions,'' \emph{IEEE Wireless Commun.}, early access, 2025, doi: 10.1109/MWC.2025.3629597.  

\bibitem{Lu2025MCOM}
W.-J. Lu {\em et al.}, ``Fluid antennas: Reshaping intrinsic properties for flexible radiation characteristics in intelligent wireless networks,'' \emph{IEEE Commun. Mag.}, vol. 63, no. 5, pp. 40--45, May 2025.

\bibitem{Meta-Fluid lot}
B. Liu {\em et al.}, ``Meta Fluid Antenna: Architecture Design, Performance Analysis, Experimental Examination,'' \emph{IEEE Internet Things J.}, early access, 2026, doi: 10.1109/JIOT.2026.3709737. 


\bibitem{Ramirez2024TWC}
P. Ramírez-Espinosa, D. Morales-Jimenez, and K.-K. Wong, ``A new spatial block-correlation model for fluid antenna systems,'' \emph{IEEE Trans. Wireless Commun.}, vol. 23, no. 11, pp. 15829--15843, Nov. 2024.

\bibitem{Zhang2025SBAR}
Z. Zhang, J. Zhu, L. Dai, and R. W. Heath, Jr., ``Successive Bayesian reconstructor for channel estimation in fluid antenna systems,'' \emph{IEEE Trans. Wireless Commun.}, vol. 24, no. 3, pp. 1992--2006, Mar. 2025.

\bibitem{New2026JSAC}
W. K. New {\em et al.}, ``Fluid antenna systems: Redefining reconfigurable wireless communications,'' \emph{IEEE J. Sel. Areas Commun.}, vol. 44, pp. 1013--1044, Nov. 2025.


\bibitem{Chai2022CL}
Z. Chai, K.-K. Wong, K.-F. Tong, Y. Chen, and Y. Zhang, ``Port selection for fluid antenna systems,'' \emph{IEEE Commun. Lett.}, vol. 26, no. 5, pp. 1180--1184, May 2022.


\bibitem{Wong2022TWC}
K.-K. Wong and K.-F. Tong, ``Fluid antenna multiple access,'' \emph{IEEE Trans. Wireless Commun.}, vol. 21, no. 7, pp. 4801--4815, Jul. 2022.

\bibitem{Wong2023SlowFAMA}
K.-K. Wong, D. Morales-Jimenez, K.-F. Tong, and C.-B. Chae, ``Slow fluid antenna multiple access,'' \emph{IEEE Trans. Commun.}, vol. 71, no. 5, pp. 2831--2846, May 2023.

\bibitem{Wong2023FastFAMA}
K.-K. Wong, K.-F. Tong, Y. Chen, and Y. Zhang, ``Fast fluid antenna multiple access enabling massive connectivity,'' \emph{IEEE Commun. Lett.}, vol. 27, no. 2, pp. 711--715, Feb. 2023.

\bibitem{Waqar2023CL}
N. Waqar, K.-K. Wong, K.-F. Tong, A. Sharples, and Y. Zhang, ``Deep learning enabled slow fluid antenna multiple access,'' \emph{IEEE Commun. Lett.}, vol. 27, no. 3, pp. 861--865, Mar. 2023.

\bibitem{Eskandari2024WCL}
M. Eskandari, A. G. Burr, K. Cumanan, and K.-K. Wong, ``cGAN-based slow fluid antenna multiple access,'' \emph{IEEE Wireless Commun. Lett.}, vol. 13, no. 10, pp. 2907--2911, Oct. 2024.

\bibitem{Wu2026CoNOMA}
T. Wu {\em et al.}, ``Unleashing more potential from FAS: A framework of FAS-CoNOMA systems,'' \emph{IEEE Trans. Commun.}, vol. 74, pp. 4820--4836, Feb. 2026.

\bibitem{Wu2026Secrecy}
T. Wu {\em et al.}, ``Variable block-correlation modeling and optimization for secrecy analysis in fluid antenna systems,'' \emph{IEEE Trans. Wireless Commun.}, vol. 25, pp. 15069--15085, Apr. 2026.

\bibitem{Wu2026ScalableFAS}
T. Wu {\em et al.}, ``Scalable fluid antenna systems: A new paradigm for array signal processing,'' \emph{IEEE J. Sel. Topics Signal Process.}, vol. 20, no. 3, pp. 389--406, Apr. 2026.

\bibitem{Wang2024AIFAS}
C. Wang, Z. Li, K.-K. Wong, R. Murch, C.-B. Chae, and S. Jin, ``AI-empowered fluid antenna systems: Opportunities, challenges, and future directions,'' \emph{IEEE Wireless Commun.}, vol. 31, no. 5, pp. 34--41, Oct. 2024.

\bibitem{Zhu2024MAopportunities}
L. Zhu, W. Ma, and R. Zhang, ``Movable antennas for wireless communication: Opportunities and challenges,'' \emph{IEEE Commun. Mag.}, vol. 62, no. 6, pp. 114--120, Jun. 2024.

\bibitem{New2025TWC}
W. K. New, K.-K. Wong, H. Xu, F. R. Ghadi, R. Murch, and C.-B. Chae, ``Channel estimation and reconstruction in fluid antenna system: Oversampling is essential,'' \emph{IEEE Trans. Wireless Commun.}, vol. 24, no. 1, pp. 309--322, Jan. 2025.

\bibitem{Xu2024CL}
H. Xu {\em et al.}, ``Channel estimation for FAS-assisted multiuser mmWave systems,'' \emph{IEEE Commun. Lett.}, vol. 28, no. 3, pp. 632--636, Mar. 2024.


\bibitem{Meta-Fluid}
D. Ma, B. Liu, J. Yang, Y. Fang, T. Wu, and K.-F. Tong, ``Meta-Fluid Antenna-Enabled Multi-User ISAC: Architecture Implementation, Algorithm Design, and Full-Wave Validation,'' \emph{IEEE J. Sel. Topics Signal Process.}, vol. 20, no. 3, pp. 247--261, Apr. 2026. 

\bibitem{Xu2025WCL}
B. Xu, Y. Chen, Q. Cui, X. Tao, and K.-K. Wong, ``Sparse Bayesian learning-based channel estimation for fluid antenna systems,'' \emph{IEEE Wireless Commun. Lett.}, vol. 14, no. 2, pp. 325--329, Feb. 2025.

\bibitem{Hassibi2003}
B. Hassibi and B. M. Hochwald, ``How much training is needed in multiple-antenna wireless links?'' \emph{IEEE Trans. Inf. Theory}, vol. 49, no. 4, pp. 951--963, Apr. 2003.

\bibitem{Biguesh2006}
M. Biguesh and A. B. Gershman, ``Training-based MIMO channel estimation: A study of estimator tradeoffs and optimal training signals,'' \emph{IEEE Trans. Signal Process.}, vol. 54, no. 3, pp. 884--893, Mar. 2006.

\bibitem{Donoho2006}
D. L. Donoho, ``Compressed sensing,'' \emph{IEEE Trans. Inf. Theory}, vol. 52, no. 4, pp. 1289--1306, Apr. 2006.

\bibitem{Tipping2001}
M. E. Tipping, ``Sparse Bayesian learning and the relevance vector machine,'' \emph{J. Mach. Learn. Res.}, vol. 1, pp. 211--244, 2001.


\bibitem{NEW-XL-MIMO}
A. Tang {\em et al.}, ``Revisiting XL-MIMO channel estimation: When dual-wideband effects meet near field,'' \emph{IEEE Trans. Wireless Commun.}, vol. 25, pp. 4231--4247, 2026.

\bibitem{NEW-MODULAR-XL}
Y. Zhang, R. Li, C. Pan, H. Ren, T. Wu, and C. Wang, ``Modular XL-array-enabled channel estimation and localization in terahertz systems,'' \emph{IEEE Trans. Cognit. Commun. Netw.}, vol. 12, pp. 3378--3392, 2026.


\bibitem{Rasmussen2006}
C. E. Rasmussen and C. K. I. Williams, \emph{Gaussian Processes for Machine Learning}. Cambridge, MA, USA: MIT Press, 2006.

\bibitem{Haibin2024WCL}
H. Zhang {\em et al.}, ``Learning-induced channel extrapolation for fluid antenna systems using asymmetric graph masked autoencoder,'' \emph{IEEE Wireless Commun. Lett.}, vol. 13, no. 6, pp. 1665--1669, Jun. 2024.

\bibitem{Arvinte2023TWC}
M. Arvinte and J. I. Tamir, ``MIMO channel estimation using score-based generative models,'' \emph{IEEE Trans. Wireless Commun.}, vol. 22, no. 6, pp. 3698--3713, Jun. 2023.

\bibitem{Chen2025TCOM}
Z. Chen, H. Shin, and A. Nallanathan, ``Generative diffusion model-based variational inference for MIMO channel estimation,'' \emph{IEEE Trans. Commun.}, vol. 73, no. 10, pp. 9254--9269, Oct. 2025.

\bibitem{Ho2020}
J. Ho, A. Jain, and P. Abbeel, ``Denoising diffusion probabilistic models,'' in \emph{Adv. Neural Inf. Process. Syst. (NeurIPS)}, pp. 6840--6851, 2020.

\bibitem{Song2020}
J. Song, C. Meng, and S. Ermon, ``Denoising diffusion implicit models,'' in \emph{Proc. ICLR}, 2021.

\bibitem{Clarke1968}
R. H. Clarke, ``A statistical theory of mobile-radio reception,'' \emph{Bell Syst. Tech. J.}, vol. 47, no. 6, pp. 957--1000, Jul.--Aug. 1968.

\bibitem{Jakes1974}
W. C. Jakes and D. C. Cox, \emph{Microwave Mobile Communications}. New York, NY, USA: Wiley, 1994.

\bibitem{Jaeckel2014}
S. Jaeckel, L. Raschkowski, K. B\"orner, and L. Thiele, ``QuaDRiGa: A 3-D multi-cell channel model with time evolution for enabling virtual field trials,'' \emph{IEEE Trans. Antennas Propag.}, vol. 62, no. 6, pp. 3242--3256, Jun. 2014.

\bibitem{Ronneberger2015}
O. Ronneberger, P. Fischer, and T. Brox, ``U-Net: Convolutional networks for biomedical image segmentation,'' in \emph{Medical Image Computing and Computer-Assisted Intervention (MICCAI)}, pp. 234--241, 2015.




 
\end{thebibliography}
\end{document}